\documentclass[11pt]{article}

\usepackage{amsmath,amssymb}
\usepackage{QED}
\usepackage{pgf}
\usepackage{tikz}
\usepackage{graphicx}
\usepackage{url}
\usepackage[boxed,linesnumbered,longend]{algorithm2e}

\newtheorem{definition}{Definition}
\newtheorem{proposition}[definition]{Proposition}
\newtheorem{lemma}[definition]{Lemma}
\newtheorem{theorem}[definition]{Theorem}
\newtheorem{corollary}[definition]{Corollary}
\newtheorem{remark}[definition]{Remark}
\newtheorem{observation}[definition]{Observation}

\newtheorem{example}[definition]{Example}

\usepackage{anysize}
\marginsize{3cm}{3cm}{3cm}{3cm}

\newcommand{\playerA}{{\sf Player~A}}
\newcommand{\playerB}{{\sf Player~B}}
\newcommand{\MoveSetA}{{\cal A}}
\newcommand{\MoveSetB}{{\cal B}}
\newcommand{\moveA}{{A}}
\newcommand{\moveB}{{B}}
\newcommand{\SafeZone}{{\sf Safe}}
\newcommand{\GoalZone}{{\sf Goal}}
\newcommand{\init}{v_0}
\newcommand{\sem}[1]{[\![ #1 ]\!]}
\newcommand{\out}{{\sf Outcome}}
\newcommand{\convh}{{\sf CH}}
\newcommand{\affh}{{\sf AH}}
\newcommand{\lins}{{\sf LS}}
\newcommand{\rank}{\mathrm{rank}}
\newcommand{\vect}{{\sf Vect}}
\newcommand{\Safe}{{\sf Safe}}
\newcommand{\uint}{{[0,1]}}
\newcommand{\setsum}{+}

\newcommand{\hide}[1]{}

\begin{document}

\title{Minkowski Games\thanks{This work was supported by the ERC inVEST (279499) project.}}

\author{St\'ephane Le Roux \and Arno Pauly \and Jean-Fran\c{c}ois Raskin}

\date{D\'epartment d'Informatique\\Universit\'e libre de Bruxelles\\Brussels, Belgium\\ \ \\ \today}

\maketitle

\begin{abstract}
We introduce and study Minkowski games. These are two player games, where the players take turns to choose positions in $\mathbb{R}^d$ based on some rules. Variants include boundedness games, where one player wants to keep the positions bounded, and the other wants to escape to infinity; as well as safety games, where one player wants to stay within a prescribed set, while the other wants to leave it.

We provide some general characterizations of which player can win such games, and explore the computational complexity of the associated decision problems. A natural representation of boundedness games yields \textrm{coNP}-completeness, whereas the safety games are undecidable.
\end{abstract}

\section{Introduction}

\paragraph{Minkowski games}
In this paper we define and study {\em Minkowski games}.  A {\em Minkowski play} is an infinite duration interaction between two players, called \playerA\/ and \playerB, in the space $\mathbb{R}^d$.  A \emph{move} in a Minkowski play is a subset of $\mathbb{R}^d$. \playerA\/ has a set of moves $\MoveSetA$ and \playerB\/ has a set of moves $\MoveSetB$. The play starts in a position $a_0 \in \mathbb{R}^d$ and is played for an infinite number of rounds as follows. For a round starting in position $a$, \playerA\/ chooses $\moveA \in \MoveSetA$ and \playerB\/ chooses a vector $b$ in $a \setsum \moveA$, where $\setsum$ denotes the Minkowski sum. Next, \playerB\/ chooses $\moveB \in \MoveSetB$ and \playerA\/ chooses a vector $a'$ in $b \setsum \moveB$. Then a new round starts in the position $a'$. The {\em outcome} of a Minkowski play is thus an infinite sequence of vectors $a_0 b_0 a_1 b_1 \dots a_n b_n \dots$ obtained during this infinite interaction. Each outcome is either winning for \playerA\/ or for \playerB, and this is specified by a winning condition.

We consider two types of winning conditions. First, we consider {\em boundedness}: an outcome $a_0 b_0 a_1 b_1 \dots a_n b_n \dots$ is winning for \playerA\/ in the boundedness game if there exists a bounded subset $\SafeZone \subseteq \mathbb{R}^d$ such that the outcomes stays in $\SafeZone$, \textit{i.e.} for all $i \geq 0$, $a_i \in \SafeZone$ and $b_i \in \SafeZone$, otherwise the play is winning for \playerB. Second, we consider {\em safety}: given a subset $\SafeZone \subseteq \mathbb{R}^d$, an outcome is winning for \playerA\/ if the outcome stays in $\SafeZone$, otherwise the play is winning for \playerB.

$\MoveSetA$ and $\MoveSetB$ could have arbitrary cardinality, but unless stated otherwise we will consider finite sets $\MoveSetA=\{\moveA_1,\moveA_2, \dots, \moveA_{n_A}\}$ and $\MoveSetB=\{\moveB_1,\moveB_2, \dots, \moveB_{n_B}\}$. Also, unless stated otherwise, both $\SafeZone$ in the safety Minkowski games and the moves in general will be bounded.

The Minkowski games are a natural mathematical abstraction to model the interaction between two agents taking actions, modeled by moves, with imprecision as the adversary resolves nondeterminism by picking a vector in the move chosen by the other player.\footnote{See further discussions on the practical appeal of these games for modeling systems evolving in multi-dimensional spaces when we report on related works.}  Perhaps more importantly, the appeal of Minkowski games comes also from their simple and natural definition. We provide in this paper general results about these games and study several of their incarnations in which moves are given as $(i)$ bounded rational polyhedral sets, $(ii)$ sets defined using the first-order theory of the reals, or $(iii)$ represented as compact or overt sets as defined in computable analysis \cite{weihrauchd}.
Note that by Borel determinacy~\cite{martin} all these games are determined, \textit{i.e.} either of the players has a strategy that is winning for sure. Our results are as follows.


\paragraph{Results}
We establish a necessary and sufficient condition for \playerA\/ to have a winning strategy in a {\em boundedness Minkowski game} with finitely many moves (Theorem~\ref{thm:bound:determinacy}) and give a simple proof (in comparison with Borel determinacy \cite{martin}) that these games are determined. We then turn our attention to computation complexity aspects of determining the winner of a game, \textit{i.e.} who has a winning strategy. The necessary and sufficient condition that we have identified leads to a {\sc coNP} solution when the moves are given as bounded rational polyhedral sets, and we provide matching lower bounds (Theorem~\ref{thm:complexityboundednesspolytope}). These results hold both for moves represented by sets of linear inequalities and moves represented as the convex hulls of a finite set of rational points. Additionally, we show that for every fixed dimension $d$, determining the winner can be done in polynomial time (Corollary~\ref{corr:fixeddimensionpolytime}). When the moves are defined using the first-order theory of the reals, determining the winner of a boundedness game is shown to be {\sc 2ExpTime}-complete (Proposition~\ref{prop:firstorder}). Finally, in the computable analysis setting, the problem is semi-decidable (Proposition~\ref{prop:computableanalysis}), and this is the best that we can hope for.

We characterize the set of winning positions in a {\em safety Minkowski game}, even with infinite $\MoveSetA$ and $\MoveSetB$, as the greatest fixed point of an operator that removes the points where \playerB\/ can provably win (in finitely many rounds). We show that this greatest fixed point can be characterized by an approximation sequence of at most $\omega$ steps (Proposition \ref{prop:omegaiteration}), for finite $\MoveSetA$ but even for infinite $\MoveSetB$. This leads to semi-decidability in the general setting of computable analysis (Proposition~\ref{prop:casafety}). Then we show that identifying the winner in a safety Minkowski game is undecidable even for finite sets of moves that are given as bounded rational polyhedral sets (Theorem~\ref{thm:complexitysafetypolytope}). As a consequence, we consider a variant of the safety Minkowski games, called  structural safety Minkowski games, where \playerA\/ must maintain safety from any starting vector within the set {\sf Safe}. We show that deciding the winner in this variant is {\sc coNP}-complete when the moves are defined as bounded rational polyhedral sets (Theorem~\ref{theo:structural}).


\paragraph{Motivations and related works}
Infinite duration games are commonly used as mathematical framework for modeling the controller synthesis problem for reactive systems~\cite{PnueliR89}. For reactive systems embedded in some physical environment, games played on hybrid automata have been considered, see e.g.~\cite{HenzingerK97a} and references therein. In such a model, one controller interacts with an environment whose physical properties are modeled by valuations of $d$ real-valued variables (vectors in $\mathbb{R}^d$). Most of the problems related to the synthesis of controllers for hybrid automata are undecidable~\cite{HenzingerK97a}. Restricted subclasses with decidable properties, such as timed automata and initialized rectangular automata have been considered~\cite{MalerPS95,HenzingerHM99}. Most of the undecidability properties of those models rely on the coexistence of continuous and discrete evolutions of the configurations of hybrid automata.  The one-sided version of the model of Minkowski games (where $\MoveSetB = \{\{0\}\}$) can be seen as a restricted form of an hybrid games in which each continuous evolution is of a unique fixed duration and space independent (such as in linear and rectangular hybrid automata). It is usually called discrete time control in this setting. While discrete-time control problems are known to be undecidable for linear hybrid automata, they are decidable for (bounded) rectangular automata~\cite{HenzingerK97a}. We show in Remark~\ref{rem:finitebisi} below how this positive result can be transferred to a subclass of Minkowski games.

To the best of our knowledge, the closest models to Minkowski games that have been considered in the literature so far are Robot games defined by Doyen et al. in~\cite{doyen} and Bounded-Rate Multi-Mode Systems introduced by Alur et al. in~\cite{AlurTW12,AlurFMT13}. Minkowski games generalize robot games: there the moves are always singletons given as integer vectors. While we show that the Safety problem is undecidable for bounded safety objectives, it is easy to show that this problem is actually decidable for robot games. However, in~\cite{doyen} they investigate reachability of a specific position rather than safety conditions as we do here. Reachability was proven undecidable in~\cite{niskanen} even for two-dimensional robot games. Boundedness objectives have not been studied for Robot games.

Bounded-Rate Multi-Mode Systems (BRMMS) are a restricted form of hybrid systems that can switch freely among a finite number of modes. The dynamics in each mode is specified by a bounded set of possible rates. The possible rates are either given by convex polytopes or as finite set of vectors. There are several differences with Minkowski games. First BRMMS are asymmetric and are thus closer to the special case of one-sided Minkowski games. Second, the actions in BRMMS are given by a mode and a time delay $\delta \in \mathbb{R}$ while the time elapsing in our model can be seen as uniform and fixed. The ability to choose delays that are as small as needed makes the safety control problem for BRMMS with modes given as polytopes decidable while we show that the safety Minkowski games with moves defined by polytopes are undecidable. The discrete time control of BRMMS, which is more similar to the safety Minkowski games, has been solved only for modes given as finite sets of vectors and left open for modes given as polytopes. Our undecidability results can be easily adapted to the discrete time control of BRMMS and thus solves the open question left in that paper. Boundedness objectives have not been studied for BRMMS.

\paragraph{Structure of the paper}

Section~\ref{sec:p} collects various basic mathematics, typically from linear algebra, that we are using in the paper. It also defines the Minkowski games. Section~\ref{sec:ag} defines and studies auxiliary games, which will be used to decompose every more complex Minkowski game into a simpler Minkowski game and a (also simpler) auxiliary game. Section~\ref{sec:grbp} characterizes the winner of a boundedness Minkowski game in terms of simple convex geometry, and it describes the winning strategies. Section~\ref{ccbp} studies the algorithmic complexity of finding the winner in the various settings. Section~\ref{wrsp} collects a few properties of the winning region of the safety problem, depending on various restrictions on the game. Section~\ref{sec:undecidable} shows, among others things, that finding the winner of a safety Minkowski game is undecidable, even for a simple subclass. In Section~\ref{sec:ssg} we consider structural safety games, and prove \textrm{coNP}-completeness for the associated decision problem. Finally, Section~\ref{sec:oq} mentions a few open questions.

\section{Preliminaries}\label{sec:p}

\paragraph{{\bf Linear constraints}} Let $d\in\mathbb{N}_{>0}$, and  $X=\{ x_1, x_2, \dots, x_d\}$ be a set of variables. A {\em linear term} on $X$ is a term of the form $\alpha_1 x_1 + \alpha_2 x_2 + \dots \alpha_d x_d$ where $x_i \in X$, $\alpha_i \in \mathbb{R}$ for all $i$, $1 \leq i \leq n$. A {\em linear constraint} is a formula $\alpha_1 x_1 + \alpha_2 x_2 + \dots \alpha_d x_d \sim c$, and $\sim \in \{ <,\leq,=,\geq,>\}$, that compares a linear term with a constant $c \in \mathbb{R}$. Given a valuation $v : X \rightarrow  \mathbb{R}$, that can be seen equivalently as a vector in $\mathbb{R}^d$, we write $v \models \alpha_1 x_1 + \alpha_2 x_2 + \dots \alpha_n x_n \sim c$ iff $\alpha_1 v(x_1) + \alpha_2 v(x_2) + \dots \alpha_d v(x_d) \sim c$. Given a linear constraint $\phi \equiv \alpha_1 x_1 + \alpha_2 x_2 + \dots \alpha_d x_d \sim c$, we write $\sem{\phi}=\{ v \in \mathbb{R}^d \mid v \models \alpha_1 x_1 + \alpha_2 x_2 + \dots \alpha_d x_d \sim c \}$. A linear constraint is rational, if all $\alpha_i$ and $c$ are rational numbers.

\paragraph{Polyhedra, polytopes, convex hull}
Given a finite set ${\cal H}=\{ \phi_1, \phi_2, \dots, \phi_n \}$ of linear constraints, we note $\sem{{\cal H}}=\{ v \in \mathbb{R}^d \mid \forall \phi \in {\cal H}: v \models \phi \}$ the set of vectors that satisfies all the linear constraints in ${\cal H}$. Such a set is a convex set and is usually called a {\em polyhedra}. In the special case that is bounded, then it is called a {\em polytope}. We call a polytope rational, if all $\phi_i$ can be chosen rational. When a polytope is {\em closed}, then it is well-known that it can be represented not only by a finite set of linear inequalities that are all non-strict but also as the {\em convex hull} of a finite set of (extremal) vectors. The convex hull of a subset of a $\mathbb{R}$-vector space is noted and defined as follows:

$$\convh(\mathcal{ V}) := \left\{ \sum_{i=0}^{n} \alpha_i x_i \,\mid\, n\in\mathbb{N} \wedge \sum_{i=0}^{n} \alpha_i = 1 \wedge \forall i(x_i \in \mathcal{V} \wedge \alpha_i \geq 0)\right\}$$

Carath{\'e}odory's theorem says that for all $\mathcal{V} \subseteq \mathbb{R}^d$, every point in $\convh(\mathcal{V})$ is a convex combination of at most $d+1$ points from $\mathcal{V}$. As a consequence, the $n$ ranging over $\mathbb{N}$ in the definition of the convex hull can safely be replaced with fixed $d$.

Let $P$ be a closed polytope. $P$ has two families of representations: its $H$-representations are the finite sets of linear inequalities ${\cal H}$ such that $\sem{{\cal H}}=P$, and its $V$-representations are the finite sets of vectors ${\cal V}$ such that $\convh({\cal V})=P$. Some algorithmic operations are easier to perform on one representation or on the other. Unfortunately, in general there cannot exist a polynomial time translation from one representation to the other unless ${\sc P}={\sc NP}$. Nevertheless, such a polynomial time translation exists for fixed dimension:

\begin{theorem}[\cite{Cha93}]
\label{thm:VH}
Let $P$ be a rational closed polytope of fixed dimension $d \in \mathbb{N}$. There exists a polynomial time algorithm that given a $H$-representation of $P$, compute a $V$-presentation of $P$, and conversely, there exists a polynomial time algorithm that given a $V$-representation of $P$, compute a $H$-presentation of $P$.
\end{theorem}

We denote by ${\sf Ver(P)}$ the extremal points, \textit{i.e.}, the vertices of a polytope $P$. It is the minimal set whose convex hull equals $P$. Note that a closed polytope is rational iff all its vertices are rational points.

\paragraph{{\bf Minkowski sum}}
For subsets $A,B \in \mathbb{R}^d$ their Minkowski sum $A+B$ is defined as $\{a+b\,\mid\,a\in A \wedge b\in B\}$. The Minkowski sum inherits commutativity and associativity from the usual sum of vectors. The set $\{0\}$ is the neutral element, but there are no inverse elements in general. If $A = \{a\}$ then $A + B$ (resp. $B + A$) is written $a + B$ (resp. $B + a$) in a slight abuse of notation. It is straightforward to prove that $\convh(A) + \convh(B) = \convh(A+B)$. Especially, if $A$ and $B$ are convex, so is $A+B$. While $A + A$ may be a strict superset of $2A := \{2a \mid a \in A\}$ in general, for convex $A$ we find $A + A = 2A$.

\paragraph{{\bf Topological closure}} The topological closure of a set $S$ is denoted by $\overline{S}$.

\paragraph{{\bf Minkowski games - Strategies}}
We have described in the introduction how the players interact in a Minkowski games by choosing in each round a move and by resolving nondeterminism among the moves chosen by the other player. We now formally define the notions of strategies for each player together with the associated compatible outcomes.

When playing Minkowski games, players are applying {\em strategies}. In a game with moves $\MoveSetA$ and $\MoveSetB$, strategies for the two players are defined as follows. A strategy for \playerA\/ is a function
$$\lambda_A : (\mathbb{R}^d)^{*} \rightarrow (\MoveSetA \cup (\mathbb{R}^d)^{*}) \times \MoveSetB \rightarrow \mathbb{R}^d$$
\noindent
that respects the following consistency constraint: for all finite sequences of positions $\rho \in (\mathbb{R}^d)^{*}$ that ends in $v \in \mathbb{R}^d$, and moves $\moveB \in \MoveSetB$, $\lambda_A(\rho,\moveB) \in v \setsum \moveB$.
Symmetrically, a strategy for \playerB\/ is a function
$$\lambda_B : (\mathbb{R}^d)^{*} \rightarrow (\MoveSetB \cup (\mathbb{R}^d)^{*}) \times \MoveSetA \rightarrow \mathbb{R}^d$$
\noindent
with the symmetric consistency constraint.

A play $a_0 b_0 a_1 b_1 \dots a_n b_n \dots$, that starts in $v_0=a_0$, is consistent with strategies $\lambda_A$ and $\lambda_B$ if for all $i \geq 0$, we have that:
$$b_i=\lambda_B(a_0 b_0 a_1 b_1 \dots a_i,\lambda_A(a_0 b_0 a_1 b_1 \dots a_i))$$
and
$$a_{i+1}=\lambda_A(a_0 b_0 a_1 b_1 \dots a_i b_i,\lambda_B(a_0 b_0 a_1 b_1 \dots a_i b_i)).$$

Given two strategies $\lambda_A$ and $\lambda_B$, one for each player, and a position $v_0 \in \mathbb{R}^d$, we note $\out_{v_0}(\lambda_A,\lambda_B)$ the unique play that starts in $v_0$ and which is consistent with the two strategies.

\paragraph{{\bf Winning conditions and variants of Minkowski games}}
By fixing the rule that determines who wins a Minkowski play, we obtain {\em Minkowski games}.
Here we consider three types of Minkowski games.

\begin{definition}
A {\em boundedness Minkowski game} is a pair $\langle \MoveSetA, \MoveSetB \rangle$ of sets of moves in $\mathbb{R}^d$ for \playerA\/ and \playerB. A play in a boundedness Minkowski game starts in some irrelevant $a_0 \in\mathbb{R}^d$, and the resulting play $a_0 b_0 a_1 b_1 \dots a_n b_n \dots$ is winning for \playerA\/ if there exists a bounded subset $\SafeZone$ of $\mathbb{R}^d$ such that $a_i,b_i \in \SafeZone$ for all $i \in \mathbb{N}$, otherwise \playerB\/ wins the game. The associated decision problem asks if  \playerA\/ has a strategy $\lambda_A$ which is winning against all the strategies $\lambda_B$ of \playerB.
\end{definition}

\begin{definition}
A {\em safety Minkowski game} is defined by $\langle \MoveSetA, \MoveSetB, \SafeZone, v_0 \rangle$, where $\MoveSetA$ and $\MoveSetB$ are sets of moves for \playerA\/ and \playerB, $\SafeZone \subseteq \mathbb{R}^d$ (bounded unless stated otherwise), and $\init \in \SafeZone$ is the initial position. A play in a safety Minkowski game starts in $\init$, and the resulting play $a_0 b_0 a_1 b_1 \dots a_n b_n \dots$ is winning for \playerA\/ if $a_i,b_i \in \SafeZone$ for all $i \in \mathbb{N}$, otherwise \playerB\/ wins the game.
The associated decision problem asks if  \playerA\/ has a strategy $\lambda_A$ which is winning against all the strategies $\lambda_B$ of \playerB.
\end{definition}

\begin{definition}
A {\em structural safety Minkowski game} is defined by $\langle \MoveSetA, \MoveSetB, \SafeZone \rangle$, where
$\MoveSetA$, and $\MoveSetB$ are sets of moves for \playerA\/ and \playerB, and $\SafeZone \subseteq \mathbb{R}^d$. In such a game, the interaction between the two player starts from any position $v_0 \in \SafeZone$, and the resulting play $a_0 b_0 a_1 b_1 \dots a_n b_n \dots$ is winning for \playerA\/ if $a_i,b_i \in \SafeZone$ for all $i \in \mathbb{N}$, otherwise \playerB\/ wins the game. The associated decision problem asks if \playerA\/ has a strategy to win the safety game wherever it starts in $\SafeZone$ and against all the strategies of \playerB.
\end{definition}

A game is {\em single-sided} if $\MoveSetB=\{\{0\}\}$, i.e. whenever \playerB\/ has only one trivial move. We use single-sided Minkowski games to show that several of our lower-bounds hold for this subclass of games.

\section{Auxiliary games}\label{sec:ag}
We will make use of two kinds of auxiliary games in proving our results on Minkowski games. These games might be of independent interest, although we only prove sufficient results for our purposes here. Our first auxiliary game captures the difference between a set and its convex hull for controlling some trajectory in $\mathbb{R}^d$. \playerB\/ plays points in some set $\convh(B) \subseteq \mathbb{R}^d$, which \playerA\/ has to approximate as well as possible while playing points in $B$.

\begin{definition}
In the convex approximation game for $B \subseteq \mathbb{R}^d$ with error margin $E \subseteq \mathbb{R}^d$, in each turn $j$ \playerB\/ plays some $v_j \in \convh(B)$, then \playerA\/ follows with some $u_j \in B$. If for all $j \in \mathbb{N}$, $\sum_{i=0}^j (v_i - u_i) \in E$, then \playerA\/ wins, else \playerB\/ wins.
\end{definition}

The precise nature of the error margins $E$ allowing \playerA\/ to win is not important for us, important is that for bounded $B$ there is some bounded $E$ enabling \playerA\/ to win.

\begin{lemma}\footnote{This result was provided by an anonymous contributor in the following answer posted to math.stackexchange.com: \url{http://math.stackexchange.com/q/1814457}.}
\label{lemma:stackexchange}
For $B \subseteq \mathbb{R}^d$ we find that $B + d\convh(B) = (d+1)\convh(B)$.
\begin{proof}
As $\convh(B)$ is convex we find that $(d+1)\convh(B)$ equivalently may denote the $d+1$-fold Minkowski sum of $\convh(B)$ with itself or the result of the product with the scalar $d+1$. Thus, the inclusion $B + d\convh(B) \subseteq (d+1)\convh(B)$ is trivial.

For the other direction, assume that $b \in (d+1)\convh(B)$. Then $(d+1)^{-1}b \in \convh(B)$. By Charath\'eodory's theorem, there are $d+1$ points $b_i \in B$ and scalars $\alpha_i \geq 0$ for $i \in \{0,\ldots,d\}$ with $\sum_{i=1}^d \alpha_i = 1$ and $\sum_{i=1}^d \alpha_i b_i = (d+1)^{-1}b$, i.e.~$b = (d+1)\sum_{i=1}^d \alpha_i b_i$. W.l.o.g.~assume that $\alpha_d \geq \alpha_i$ for all $i \leq d$. Then in particular $\alpha_d \geq (d+1)^{-1}$. Now we can write: \[b = b_d + d \left [ \frac{(d+1)\alpha_d - 1}{d}b_d + \sum_{i=0}^{d-1} \frac{d+1}{d}\alpha_i b_i \right ]\]
The expression in square brackets is a convex combination of the $b_i$, thus we can conclude $b \in B + d\convh(B)$.
\end{proof}
\end{lemma}

\begin{proposition}
\label{prop:convexapproximation}
Pick $c \in d\cdot\convh(B)$. \playerA\/ has a winning strategy in the convex approximation game for $B \subseteq \mathbb{R}^d$ with error margin $d\cdot\convh(B) + \{-c\}$.
\begin{proof}
We describe a strategy of \playerA\/ that ensures $\sum_{i=0}^j (v_i - u_i) \in d\convh(B)$ inductively. The case $t = 0$ is satisfied since $0 \in d\convh(B) + \{-c\}$ by choice of $c$. By induction hypothesis, we find that $v_j + \sum_{i=0}^{j-1} (v_i - u_i) + c \in \convh(B) + d\convh(B)$. By Lemma \ref{lemma:stackexchange}, there exists some $u_j \in B$ and $r \in d\convh(B)$ such that  $v_j + \sum_{i=0}^{j-1} (v_i - u_i) + c = u_j + r$, i.e.~$r - c = \sum_{i=0}^j (v_i - u_i) \in d\convh(B) + \{-c\}$ as desired.
\end{proof}
\end{proposition}

If we place some restrictions on the set $B$, we can obtain better bounds. Essentially, our condition is that $B$ contains the boundary of $\convh(B)$.

\begin{proposition}
Let $B$ be closed and satisfy $d(x,B) = d(x,\convh(B))$ for each $x \notin \convh(B)$. Let $\rho := \max_{x \in \convh(B)} d(x,B)$ and $\overline{B}(0,\rho) \subseteq E$. \playerA\/ has a winning strategy for the convex approximation game for $B$ with error margin $E$.
\begin{proof}
We describe a strategy of \playerA\/ that ensures $\sum_{i=0}^j (v_i - u_i) \in \overline{B}(0,\rho)$ inductively. The case $j = 0$ is trivially satisfied. Assume that $\sum_{i=0}^{j-1} (v_i - u_i) \in \overline{B}(0,\rho)$, and that \playerB\/ chose $v_j$ in round $j$. \playerA\/ will play some $u_j \in B$ with $d(\sum_{i=0}^{j-1} (v_i - u_i) + v_j,u_j) = d(\sum_{i=0}^{j-1} (v_i - u_i) + v_j,B)$.

It remains to show that $\sum_{i=0}^{j} (v_i - u_i) \in \overline{B}(0,\rho)$, i.e.~that $d(\sum_{i=0}^{j-1} (v_i - u_i) + v_j,u_j) \leq \rho$. If $\sum_{i=0}^{j-1} (v_i - u_i) + v_j \in \convh(B)$, this is true by definition of $\rho$. Else, by our various assumptions: \begin{align*} d(\sum_{i=0}^{j-1} (v_i - u_i) + v_j,u_j) & = d(\sum_{i=0}^{j-1} (v_i - u_i) + v_j, B)\\ & = d(\sum_{i=0}^{j-1} (v_i - u_i) + v_j,\convh(B)) \\ & \leq d(\sum_{i=0}^{j-1} (v_i - u_i) + v_j, v_j) \\ & = d(\sum_{i=0}^{j-1} (v_i - u_i), 0) \leq \rho\end{align*}
\end{proof}
\end{proposition}

Our second class of auxiliary games is (up to some details) a special case of the convex approximation games:

\begin{definition}
In the $d$-dimensional $+1/-1$-game with threshold $r$, \playerA\/ plays positions $n_i \in \{1,\ldots,d\}$ and \playerB\/ plays stochastic $d$-dimensional vectors $v_i \in \mathbb{S}^d$. \playerA\/ wins if for all $j \in \mathbb{N}$ and all $k \in \{1,\ldots,d\}$ we find that $|\{i \leq j \mid n_i = k\}| - \sum_{i = 0}^j (v_i)_k \geq r$.
\end{definition}

This means that in each round \playerA\/ is putting a unit token on one out of $d$ positions, while \playerB\/ is removing fractions of token summing up to one unit from the positions. \playerB\/ attempts to get some position below $f$, \playerA\ wants to prevent this.

\begin{proposition}
\label{prop:plusoneminusone}
\playerA\/ has a winning strategy in the $d$-dimensional $+1/-1$-game with threshold $-d$.
\begin{proof}
Consider the convex approximation game for: \[B := \{(1,0,\ldots,0), (0,1,\ldots,0),\ldots,(0,0,\ldots,1)\}\] Pick $c = (1,1,\ldots,1)$. Then by Proposition \ref{prop:convexapproximation} \playerA\/ can force all components of the vectors $\sum_{i=0}^j (v_i - u_i)$ to not exceed $d-1$. Since in a $+1/-1$-game, \playerA\/ moves first, whereas in a convex approximation game \playerB\/ moves first, we need to adjust this bounds by adding the maximum deviation possible through a single round: $1$.
\end{proof}
\end{proposition}

\begin{proposition}
\playerB\/ has a winning strategy in the $d$-dimensional $+1/-1$-game with threshold $-H(d-1)$ (the $d-1$-th harmonic number).
\begin{proof}
For $k < d$, in the $k$-th round there are at least $(d - k)$ positions never played by \playerA\/ so far. \playerB\/ plays each of these with weight $\frac{1}{d - k}$ each. In round $(d - 1)$, this gives total weight $-\sum_{i=1}^{d-1} i^{-1} = -H(d-1)$ to the position never played by \playerA\/.
\end{proof}
\end{proposition}

We leave the question open to precisely determine the following values: \[\tau_d = - \sup \{r \mid \textnormal{\playerA\/ wins the } +1/-1\textnormal{-game with threshold } r\}\] What we will use is only that there are thresholds allowing \playerA\/ to win.

\section{General results on the boundness problem}\label{sec:grbp}
To start this section, we consider the special case of one-sided boundedness Minkowski games and provide a sufficient (and necessary) condition for \playerA\/ to win. The proof showcases some ideas we will then use to fully characterize the general case. The characterization in the general case in particular implies that the condition for the one-sided case is necessary.

\begin{proposition}
We consider a one-sided boundedness Minkowski game $\langle \MoveSetA, \{0\} \rangle$ where $\MoveSetA = \{A_1, \ldots, A_n\}$ and such that $0 \in \convh((x_i)_{1 \leq i \leq n})$ for all tuples $(x_i)_{1 \leq i \leq n}$ in $A_1 \times \dots \times A_n$. Then \playerA\/ wins the boundedness game.
\begin{proof}
We describe the current state by some list of pairs $(x_i,\alpha_i)_{i \leq n}$ such that $x_i \in A_i$ and $\alpha_i \in [0,1]$. We keep two invariants satisfied throughout the play: First, it will always be the case that the current position is equal to $\sum_{i \leq n} \alpha_i x_i$, which by boundedness of each $A_i$ implies that \playerA\/ wins. Second, we maintain the invariant that there is some $k \leq n$ with $\alpha_{k} = 0$. Initially, the choice of the $x_i$ is arbitrary, and all $\alpha_i$ are $0$. This ensures that the strategy we describe for \playerA\/ is well-defined.

On his turn, \playerA\/ plays some $A_{k}$ for $k$ with $\alpha_{k} = 0$. \playerB\/ reacts with some $x'_k \in A_k$, and we set $x_k := x'_k$ and $\alpha_k := 1$.

 If immediately after the move, no $\alpha_i$ is currently $0$, we write a convex combination $0 = \sum_{i \leq n} \beta_i x_i$, which is possible by assumption. Let $r := \max_{i \leq n} \frac{\beta_i}{\alpha_i}$, and then update $\alpha_i = \alpha_i - r^{-1}\beta_i$. By the choice of the $\beta_i$, this leaves $\sum_{i \leq n} \alpha_i x_i$ unchanged. The choice of $r$ ensures that $\alpha_i \in [0,1]$ remains true, and more over, after the updating process, there is some $k \leq n$ with $\alpha_k = 0$. Thus, the invariant is true again after the updating process.
\end{proof}
\end{proposition}

We introduce some notation to formulate the main lemma of this section,which is then summarized by Theorem~\ref{thm:bound:determinacy}. For some set of moves $\MoveSetB$ let $\convh(\MoveSetB) := \{\convh(B) \mid B \in \MoveSetB\}$ and $\overline{\MoveSetB} := \{\overline{B} \mid B \in \MoveSetB\}$. We say that a strategy for \playerB\/ in a Minkowski game is \emph{simple}, if it prescribes choosing always the same $B \in \MoveSetB$, and if the choice $a_i \in A_i$ depends only on the choice of $A_i \in \MoveSetA$ by \playerA.

\begin{theorem}
\label{thm:bound:determinacy}
\begin{itemize}
\item Boundedness Minkowski games are determined;
\item the winner is the same for $\langle \MoveSetA, \MoveSetB \rangle$ and $\langle \convh(\overline{\MoveSetA}), \convh(\overline{\MoveSetB}) \rangle$;
\item if \playerB\/ has a winning strategy, she has a simple one;
\item \playerA\/ wins iff $0 \in (\convh \{a_i \mid i \leq n\}) + \convh(\overline{B})$ for all $(a_i)_{i \leq n}$ with $a_i \in A_i$ and $B \in \MoveSetB$.
\end{itemize}
\begin{proof}
The claims follow from Lemma \ref{lem:bound} below.\end{proof}
\end{theorem}

Note that the determinacy of the boundedness Minkowski games also follows from Borel determinacy~\cite{martin} (and also from techniques in \cite{Wolfe55}) since the set of the winning plays for \playerA\/ is a $\Sigma^0_2$ set (for the usual product topology over discrete topology).

\begin{lemma}
\label{lem:bound}
The following are equivalent for a boundedness Minkowski game $\langle \MoveSetA = \{A_1,\ldots,A_n\}, \MoveSetB \rangle$:
\begin{enumerate}
\item \playerA\/ has a winning strategy in $\langle \MoveSetA, \MoveSetB \rangle$.
\item \playerA\/ has a winning strategy in $\langle \MoveSetA, \convh(\MoveSetB) \rangle$.
\item \playerA\/ has a winning strategy in $\langle \MoveSetA, \overline{\MoveSetB} \rangle$.
\item \playerB\/ has no winning strategy in $\langle \MoveSetA, \MoveSetB \rangle$.
\item \playerB\/ has no simple winning strategy in $\langle \MoveSetA, \MoveSetB \rangle$.
\item \playerB\/ has no winning strategy in $\langle \convh(\MoveSetA), \MoveSetB \rangle$.
\item \playerB\/ has no winning strategy in $\langle \overline{\MoveSetA}, \MoveSetB \rangle$.
\item For all $(a_i)_{i \leq n}$ with $a_i \in \convh(\overline{A}_i)$ and $B \in \MoveSetB$ we find that: \[0 \in (\convh \{a_i \mid i \leq n\}) + \convh(\overline{B})\]
\item For all $(a_i)_{i \leq n}$ with $a_i \in A_i$ and $B \in \MoveSetB$ we find that: \[0 \in (\convh \{a_i \mid i \leq n\}) + \convh(\overline{B})\]
\end{enumerate}
\begin{proof}
That $1.$ and $2.$ are equivalent is shown in Lemma \ref{lem:bound:conv} below. That $1.$ and $3.$ are equivalent is shown in Lemma \ref{lem:bound:cl} below. The implication from $1.$ to $4.$ is trivial, so is $4. \Rightarrow 5.$. That $\neg 9.$ implies $\neg 5.$ is Lemma \ref{lem:bound:playerb}. Using the equivalences of $1.$, $2.$ and $3.$ we see that it suffices to show that $8.$ implies $1.$ in the special case where all $B \in \MoveSetB$ are closed and convex. This is the statement of Lemma \ref{lem:bound:playera}. That $8.$ implies $9.$ is trivial, that $9.$ implies $8.$ is shown in Lemma \ref{lem:bounds:hyperplane}. That $4.$,$6.$,$7.$ are mutually equivalent then follows.
\end{proof}
\end{lemma}

\begin{lemma}
\label{lem:bound:conv}
\playerA\/ has a winning strategy in $\langle \MoveSetA, \MoveSetB \rangle$ iff
she has a winning strategy in $\langle \MoveSetA, \MoveSetB_{\textrm{conv}} \rangle$.
\begin{proof}
Every strategy for \playerA\/ in $\langle \MoveSetA, \MoveSetB\rangle$ is also a valid strategy in $\langle \MoveSetA, \MoveSetB_{\textrm{conv}}\rangle$, and if the strategy is winning in the former game, it is winning in the latter. Thus, we only need to show how to transform a winning strategy $s'$ for \playerA\/ in $\langle \MoveSetA, \MoveSetB_{\textrm{conv}}\rangle$ to a winning strategy $s$ for \playerA\/ in $\langle \MoveSetA, \MoveSetB\rangle$.

We do this by using an auxiliary convex approximation game for each $B \in \MoveSetB$. In the convex approximation game for $B$, \playerA\/ will consider the choices $y' \in \convh(B)$ prescribed to her by the strategy $s'$ in $\langle \MoveSetA, \MoveSetB_{\textrm{conv}}\rangle$ as the moves of her opponent, and will determine moves $y \in B$ according to some winning strategy for some suitable bounded set $E_B$ (which she has by Proposition \ref{prop:convexapproximation}). The strategy $s$ now chooses $y$ for $\langle \MoveSetA, \MoveSetB\rangle$.

If $s'$ enforces that the play remains within some set $E$, then by linearity, $s$ enforces that the play remains within $E + \bigoplus_{B \in \MoveSetB} E_B$, which by finiteness of $\MoveSetB$ is again a bounded set.
\end{proof}
\end{lemma}

\begin{lemma}
\label{lem:bound:cl}
\playerA\/ has a winning strategy in $\langle \MoveSetA, \MoveSetB\rangle$ iff
she has a winning strategy in $\langle \MoveSetA, \MoveSetB_{\textrm{cl}}\rangle$.
\begin{proof}
Every strategy for \playerA\/ in $\langle \MoveSetA, \MoveSetB\rangle$ is also a valid strategy in $\langle \MoveSetA, \MoveSetB_{\textrm{cl}}\rangle$, and if the strategy is winning in the former game, it is winning in the latter. Thus, we only need to show how to transform a winning strategy $s'$ for \playerA\/ in $\langle \MoveSetA, \MoveSetB_{\textrm{cl}}\rangle$ to a winning strategy $s$ for \playerA\/ in $\langle \MoveSetA, \MoveSetB\rangle$.

For this, let $s$ agree with $s'$ on which moves $A \in \mathcal{A}$ \playerA\/ is choosing, and be such that if $s'$ picks some $y' \in \overline{B} \in \MoveSetB_{\textrm{cl}}$ in round $n$, then $s$ picks some $y \in B$ with $d(y,y') < 2^{-n}$. As the notion of a play is linear, if $s'$ enforces that the play stays within some set $E$, then $s$ enforces that the play stays within $E + \overline{B}(0,2)$.
\end{proof}
\end{lemma}

\begin{lemma}
\label{lem:bound:playerb}
Consider a boundedness Minkowski game $\langle \MoveSetA = \{a_1,\ldots,a_n\}, \MoveSetB\rangle$ such that there are $a_i \in A_i$, $B \in \MoveSetB$ such that: \[0 \notin (\convh \{a_i \mid i \leq n\}) + \convh(\overline{B})\]
Then \playerB\/ has a simple winning strategy by always choosing the witnesses.
\begin{proof}
Let $u$ be the convex projection of $0$ onto $(\convh \{a_i \mid i \leq n\}) + \convh(\overline{B})$. After each round, the position will move by $|u|$ in direction $u$, hence the play will diverge.
\end{proof}
\end{lemma}

\begin{lemma}
\label{lem:bound:playera}
Consider a boundedness Minkowski game $\langle \MoveSetA = \{A_1,\ldots,A_n\}, \MoveSetB\rangle$ such that every $B \in \MoveSetB$ is closed and convex, and for all $a_i \in \convh(A_i)$, $B \in \MoveSetB$ we find that: \[0 \in (\convh \{a_i \mid i \leq n\}) + B\]
Then \playerA\/ has a winning strategy.
\begin{proof}
We will reduce the boundedness Minkowski game satisfying these conditions to a $+1/-1$-game. Central to the reduction is that we can describe the current position in the boundedness Minkowski game in the form $x_1a_1 + \ldots + x_na_n + p$ with $x_i \geq 0$, $a_i \in \convh(A_i)$ and $p$ being some fixed vector. Initially, we choose $x_i = n$, the $a_i \in \convh(A_i)$ arbitrarily, and $p$ such that the resulting expression equals $v_0$. The values $x_i$ will be considered as the positions in the $+1/-1$-game.

If \playerB\/ picks some $a'_i \in A_i$, then $x_ia_i + a'_i = (x_i + 1)\left [ \frac{x_i}{x_i+1}a_i + \frac{1}{x_i+1}a'_i\right ]$, which the expression within $[ \ ]$ being an element of $\convh(A_i)$. Thus, the choice of $A_i$ by \playerA\/ can be considered as choosing the $i$-th position in the $+1/-1$-game.

Given some move $B \in \MoveSetB$ and the current value of the $a_i$, we know that there is some $b \in B$ with $b = -\sum_{i=1}^n\alpha_i a_i$, where $\alpha_i \geq 0$, $\sum_{i=1}^n \alpha_i = 1$. \playerA\/ will choose such a $b$, which corresponds to updating the $x_i$ to $x_i - \alpha_i$. Thus, the choice by \playerB\/ can be seen as \playerB\/ making a move in the $+1/-1$-game\footnote{\playerB\/ can not induce all moves available to her in the $+1/-1$-game by picking some $B \in \MoveSetB$, but this is irrelevant for our purpose, as we are concerned with winning strategies of \playerA.}.

By Proposition \ref{prop:plusoneminusone}, \playerA\/ has a strategy in the $+1/-1$-game that ensures that $x_i \geq 0$ remains true. This in turn implies that the play in the Minkowski game remains bounded, i.e.~\playerA\/ wins.
\end{proof}
\end{lemma}

\begin{lemma}
\label{lem:bounds:hyperplane}
Let $B$ be compact and convex. If there are $a_i \in \convh(\overline{A}_i)$, $i \leq n$ such that $0 \notin \convh \{a_1,\ldots,a_n\} + B$, then there are $a'_i \in A_i$ with $0 \notin (\convh \{a'_1,\ldots,a'_n\}) + B$.
\begin{proof}
First, note that $0 \notin (\convh \{a_1,\ldots,a_n\}) + B$ is an open property in the $a_i$, hence it remains true under small perturbations of the $a_i$. Thus, replacing $A_i$ with $\overline{A}_i$ does not change anything.

Second, $0 \notin (\convh \{a_1,\ldots,a_n\}) + B$ is equivalent to $\convh \{a_1,\ldots,a_n\} \cap -B = \emptyset$. It is known that in $\mathbb{R}^d$, disjoint compact convex sets are separated by hyperplanes. Let $P$ be a hyperplane separating $\mathbb{R}^d$ into $L \supseteq \convh \{a_1,\ldots,a_n\}$ and $U \supseteq -B$. Now since $a_i \in L \cap \convh(A_i)$, we can conclude that $L \cap A_i \neq \emptyset$. Pick $a'_i \in L \cap A_i$. Then $L \supseteq \convh \{a'_1,\ldots,a'_n\}$, hence $0 \notin (\convh \{a'_1,\ldots,a'_n\} ) + B$.
\end{proof}
\end{lemma}

\section{Computational complexity of the boundedness problem}\label{ccbp}

In the previous section, we have provided general results on boundedness Minkowski games. Here we study the computational complexity of the associated decision problem~\footnote{For all our complexity results, all the encoding of numbers and vectors that we use are the natural ones, i.e. integer or rational numbers are encoded succinctly in binary.}.
To formulate complexity results, we need to consider classes of games that are defined in some effective way. We consider here three ways for the representation of sets of moves: by finite sets of linear constraints (or the convex hull of a finite set of vectors), by formulas in the first-order theory of the reals (that strictly extend the expressive power of linear constraints), and as compact sets or overt sets (closed sets with positive information) in the sense of computable analysis.

\subsection{Moves defined by linear constraints or as convex hulls}
\label{subsec:linconstraints}
We prove the following main result in this section:

\begin{theorem}
\label{thm:complexityboundednesspolytope}
Given a boundedness Minkowski game $\langle \MoveSetA, \MoveSetB \rangle$ where moves in the sets of moves $\MoveSetA$ and $\MoveSetB$ are defined by finite sets of rational linear constraints or as convex hulls of a finite sets of rational vectors, deciding the winner is {\sc coNP-complete}. The hardness already holds for one-sided boundedness games.
\end{theorem}

We establish this result by showing how to reduce the 3-SAT problem to the complement of the boundedness Minkowski game problem. For that we need some intermediate results.
A simple strategy $\lambda_B$ for \playerB\/ is called a \emph{vertex strategy}, if the $a_i \in A_i$ chosen by $\lambda_B$ are always some vertex of $A_i$.

\begin{corollary}
\label{corr:vertexStrat}
If \playerB\/ has a winning strategy in a boundedness Minkowski game $\langle \MoveSetA, \MoveSetB \rangle$ with closed moves in $\MoveSetA$ then she has a winning vertex strategy.
\begin{proof}
By Lemma~\ref{lem:bound} $(4. \Leftrightarrow 6.)$, \playerB\/ wins $\langle \MoveSetA, \MoveSetB \rangle$ iff she wins $\langle \{{\sf Ver}(A) \mid A \in \MoveSetA\}, \MoveSetB, a_0 \rangle$. By Lemma~\ref{lem:bound} $(\neg 4. \Rightarrow \neg 5.)$ applied to the latter game, she then even has a simple winning strategy in $\langle \{\textrm{Ver}(A) \mid A \in \MoveSetA\}, \MoveSetB, a_0 \rangle$. But this is just the definition of a vertex strategy.
\end{proof}
\end{corollary}

As a consequence of the previous corollary and of the determinacy of boundedness Minkowski games (Corollary \ref{thm:bound:determinacy}), to show that \playerA\/ has a winning strategy, it is sufficient to show that he can spoil all the vertex strategies of \playerB. This is an important ingredient of the reduction below.

\begin{lemma} There is a polynomial time reduction from the 3SAT problem to the complement of the boundedness problem for one-sided Minkowski games with moves defined by closed polytopes.
\end{lemma}
\begin{proof}
First of all, let us point out that the proof that we provide below works for both the $H$-representation and the $V$-representation. This is because the moves that we need to construct are all the convex hull of exactly three vectors. So the $H$-representation of such a convex hull can be obtained in polynomial time.

Let $\Psi=\{C_1,C_2,\dots,C_n\}$ be a set of clauses with 3 literals defined on the set of Boolean variables $X=\{ x_1, x_2, \dots, x_m\}$. Each $C_i$ is of the form $\ell_{i1} \lor \ell_{i2} \lor \ell_{i3}$ where each $\ell_{ij}$ is either $x$ or $\neg x$ with $x \in X$.

To define the set of moves $\MoveSetA$ for \playerA\/, we associate a move $\moveA_i$ with each clause $C_i$. The move is a subset of $\mathbb{R}^d$, where $d=2 \cdot |X|=2 \cdot m$, defined from $C_i$ as follows. We associate with each variable $x_k \in X$ two dimensions of $\mathbb{R}^d$: $2k-1$ and $2k$, and to each literal $\ell_{ij}$ a vector noted $\vect(\ell_{ij})$ defined as follows. If the literal $\ell_{ij}=x_k$, then the vector $\vect(\ell_{ij})$ has zeros everywhere but in dimension $2k-1$ and $2k$ where it is respectively equal to $1$ and $-1$. If the literal $\ell_{ij}=\neg x_k$, then the vector $\vect(\ell_{ij})$ has zeros everywhere but in dimension $2k-1$ and $2k$ where it is respectively equal to $-1$ and $1$. So, for all literals $\ell_1$ and $\ell_2$, $\vect(\ell_1)+\vect(\ell_2)={\bf 0}$ if and only if $\ell_1 \equiv \neg \ell_2$ or $\ell_2 \equiv \neg \ell_1$. Finally, the move associated with the clause $C_i=\ell_{i1} \lor \ell_{i2} \lor \ell_{i3}$ is
$$\moveA_i=\convh(\vect(\ell_{i1}),\vect(\ell_{i2}),\vect(\ell_{i3})).$$

It remains to prove the correctness of our reduction. By Corollary~\ref{corr:vertexStrat}, \playerB\/ has a winning strategy iff she has a winning vertex strategy, so we only need to consider the latter. We call a vertex strategy $\lambda_B^v$ of \playerB\/ \emph{valid}, iff there are no $i_1$, $i_2$ such that $\lambda_B^v(\moveA_{i_1})=-\lambda_B^v(\moveA_{i_2})$. We will argue first that \playerB\/ has a valid vertex strategy iff there is a satisfying assignment for $\Psi$. Then we argue that a vertex strategy for \playerB\/ is winning iff it is valid. The two parts together with Corollary~\ref{corr:vertexStrat} yield the desired claim that \playerB\/ has a winning strategy iff $\Psi$ is satisfiable.

{\bf Claim:} There is a valid vertex strategy iff there $\Psi$ is satisfiable.

Given a satisfying truth assignment, we can pick some vertex strategy such that the vertex chosen always corresponds to some true literal. In particular, we never chose vertices corresponding to both $x$ and $\neg x$ -- but by construction of the moves, this ensures that the vertex strategy is valid. Conversely, a valid vertex strategy is never choosing vertices corresponding to both $x$ and $\neg x$. Thus, we can obtain a truth assignment by making all literals corresponding to vertices chosen by the strategy true, and choosing arbitrarily for the remaining literals. This truth assignment satisfies $\Psi$ by construction.

{\bf Claim:} A vertex strategy is valid iff it is winning.

Assume that a vertex strategy $\lambda_B^v$ is not valid, and that moves $\moveA_{i_1}$, $\moveA_{i_2}$ witness this. Then if \playerA\/ alternates between playing $\moveA_{i_1}$ and $\moveA_{i_2}$, the resulting game remains bounded and is hence won by \playerA. If \playerB\/ plays some valid vertex strategy $\lambda_B^v$, and \playerA\/ plays infinitely often some move $\moveA_i$, then the position will diverge in the two dimensions associated with the literal $\lambda_B^v(\moveA_i)$. By the pigeon hole principle, \playerA\/ has to play some move infinitely often. It follows that a valid vertex strategy is a winning strategy.
\end{proof}

We have now established the hardness part of Theorem~\ref{thm:complexityboundednesspolytope}. The {\sc coNP}-membership part is covered by the following lemma.

\begin{lemma}
Negative instances of the boundedness Minkowski games expressed with moves defined as sets of rational linear inequalities or convex hull of finite sets of rational vectors can be recognized by a nondeterministic polynomial time Turing machine.
\end{lemma}
\begin{proof}
To show that \playerA\/ has no winning strategy, by Lemma~\ref{lem:bound} $(1. \Leftrightarrow 8. \Leftrightarrow 9.)$, it suffices to exhibit $a_1 \in {\sf Ver}(\moveA_1)$, $a_2 \in {\sf Ver}(\moveA_2)$, \dots, $a_{n_A} \in {\sf Ver}(\moveA_{n_A})$, and one $B \in \MoveSetB$, such that
$$  {\bf 0} \not\in \convh(a_1,a_2,\dots,a_{n_A}) \setsum \overline{\moveB}.$$

If each $\moveA_i$ is given by a set of linear constraints, each vertex in ${\sf Ver}(\moveA_i)$ has a binary representation which is polynomial in the description of $\moveA_i$, and so those points can be guessed in polynomial time. If the $\moveA_i$ are given as convex hulls of a finite set of points, we can obviously guess a vertex in polynomial time, too.

Finally, let us show that we can check in deterministic polynomial time that
$$ {\bf 0} \not\in \convh(a_1,a_2,\dots,a_{n_A}) \setsum \overline{\moveB}.$$
If $B$ is given via rational linear inequalities, this is equivalent to decide if the following set of linear constraints is unsatisfiable:
$$
\begin{array}{l}
  \bigwedge_{i=1}^{i=n_A} 0 \leq \alpha_i \leq 1 \\
  \land \sum_{i=1}^{i=n_A} \alpha_i = 1 \\
  \land x=\sum_{i=1}^{i=n_A} \alpha_i a_i \\
  \land y \in \overline{\moveB} \\
  \land 0 =x+y
\end{array}
$$

If $B$ is given as $\convh(b_1,\ldots,b_k)$ with rational $b_j$, then we need to decide whether: \[{\bf 0} \notin \convh(a_1 + b_1,a_2 + b_1,\dots,a_{n_A} + b_k)\]

Thus, the problem reduces to deciding feasibility of rational systems of inequalities, which is known to be decidable in polynomial time.
\qed
\end{proof}

\subsection{Fixed dimension and polytopic moves}
\label{subsec:fixeddim}

This section shows that given $d \in \mathbb{N}$ and a Minkowski game $\langle \MoveSetA, \MoveSetB \rangle$ with closed polytopic moves in $\mathbb{R}^d$, deciding which player has a winning strategy can be done in deterministic polynomial time. Note that for a fixed $d$ we can translate $V$-representations of (closed) polytopes into $H$-representations, and vice-versa (see Theorem~\ref{thm:VH}), so w.l.o.g. we focus here on the $V$-representation of polytopes. In this setting, by Lemma~\ref{lem:bound} it suffices to consider games with finite moves, since the extremal points of a polytope are finitely many. The degree of the polynomial (upper-)bounding the algorithmic complexity will be $2d+2$ in general, and $d+1$ for single-sided games. By Lemma~\ref{lem:bound} again, \playerB\/ has a winning strategy iff there exist $a_1 \in A_1,\dots, a_n \in A_n$ (the moves in $\MoveSetA$) and $B \in \MoveSetB$ such that $0 \notin \convh(a_1,\dots, a_n) \setsum \convh(B)$. Trying out all the tuples $(a_1,a_2,\dots,a_n)$ cannot be done in polynomial time. Instead let us rephrase the condition using a hyperplane separation result.

\begin{observation}\label{obs:hyp-sep-eq-win}
\playerB\/ has a winning strategy iff there exist $a_1 \in A_1,\dots, a_n \in A_n$, $B \in \MoveSetB$, and an hyperplane separating $\{a_1,\dots, a_n\} \setsum B$ from $0$.
\end{observation}

Trying out all the infinitely many hyperplanes is also unfeasible, so we will show how to restrict the search space to a small finite set of hyperplanes. Let us first give a very rough intuition. If a separating hyperplane exists, we can "push" it away from $0$ while retaining its separating property. Once a critical point (which cannot be passed without losing the property) is hit, the hyperplane can still be rotated around the point (or axis of two points, \textit{etc.}). Ideally, the hyperplane would eventually settle while containing $d$ affinely independent points from $\cup\MoveSetA \setsum B$, so it would suffice to check all the possible settling positions. There is a difficulty, though: rotating may go on and on without ever settling if $\cup\MoveSetA \setsum B$ does not contain $d$ affinely independent points. This difficulty is overcome by adding finitely many "dummy" points to $\cup\MoveSetA \setsum B$. Typically, adding the canonical basis will do just fine.

\begin{theorem}\label{thm:minkowsky-single-hyperplane}
Let $\langle \MoveSetA, \MoveSetB \rangle$ be a Minkowski game with finite moves in $\mathbb{R}^d$, and let $C = \{e_1,\dots, e_d\}$ be the canonical basis of $\mathbb{R}^d$. The game is won by \playerB\/ iff there exist $B \in \MoveSetB$ and affinely independent $x_1,\dots,x_d \in (\cup \MoveSetA \setsum B) \cup C$ s.t. for all $A \in \MoveSetA$ there exists $a \in A$ s.t. the affine hull of $x_1, \dots, x_d$ separates $a \setsum B$ from $0$.
\end{theorem}

Before proving Theorem~\ref{thm:minkowsky-single-hyperplane} we recall some basics in linear/affine algebra and prove a few lemmas. For all $E \subseteq \mathbb{R}^d$ let $\lins(E) := \{\sum_{i = 0}^n\alpha_ix_i\,\mid\, \forall i(x_i\in E \wedge \alpha_i\in \mathbb{R})\}$ be the linear span of $E$, and let $\affh(E) := \{\sum_{i = 0}^n\alpha_ix_i\,\mid\, \forall i(x_i\in E \wedge \alpha_i\in \mathbb{R}) \wedge \sum_{i=0}^n\alpha_i = 1\}$ be its affine hull. Clearly $\convh(E) \subseteq \affh(E) \subseteq \lins(E)$ for all subsets $E \subseteq \mathbb{R}^d$. A set of points is called affinely independent if no point is in the affine hull of the others. Affine hulls are affine spaces, \textit{i.e.} sums $x + L$ where $x\in \mathbb{R}^d$ and $L$ is a linear subspace of $\mathbb{R}^d$. Also, $x+L = y + L$ for all $y\in x+L$. In particular, $\affh(E) = x + \lins(E-x)$ for all $x \in E$. The dimension of $x + L$ is defined as the dimension of $L$, so for all $x_1,\dots,x_n \in \mathbb{R}^d$
\begin{align}
\dim \affh(x_1,\dots,x_n) = \rank(x_2 - x_1,\dots,x_n - x_1) \label{eq:rk-ah}
\end{align}
and $\dim \convh(x_1,\dots,x_n) := \dim \affh(x_1,\dots,x_n)$. It is straightforward to show
\begin{align}
\lins(x_1,\dots,x_n) = \affh(0,x_1, \dots, x_n) \label{eq:ls-ah}
\end{align}

Observation~\ref{obs:basic-la} states two inequalities about the rank of a finite set of vectors, and Lemma~\ref{lem:rank-equal} characterizes when the second one is an equality.

\begin{observation}\label{obs:basic-la}
$\rank(x_1,\dots,x_n) \leq 1 + \rank(x_2 - x_1,\dots,x_n - x_1) \leq 1 + \rank(x_1,\dots,x_n)$.
\begin{proof}
First inequality: if $x_1,\dots,x_k$ are linearly independent, so are $x_2 - x_1,\dots, x_k - x_1$. Second one:  $\affh(x_1, \dots, x_n) \subseteq \affh(0,x_1, \dots, x_n) = \lins(x_1,\dots,x_n)$ by (\ref{eq:ls-ah}).
\end{proof}
\end{observation}

\begin{lemma}\label{lem:rank-equal}
$0 \in \affh(x_1,\dots,x_n)$ iff $\rank(x_1,\dots,x_n) = \rank(x_2 - x_1,\dots,x_n - x_1)$.
\begin{proof}
$0 \in \affh(x_1,\dots,x_n)$ iff $\affh(x_1,\dots,x_n) = \affh(0,x_1,\dots,x_n)$, \textit{i.e.} iff $ \affh(x_1,\dots,x_n) = \lins(x_1,\dots,x_n)$ by (\ref{eq:ls-ah}), \textit{i.e.} iff $\rank(x_2 - x_1,\dots,x_n - x_1) = \rank(x_1,\dots,x_n)$ by (\ref{eq:rk-ah}).
\end{proof}
\end{lemma}

Lemma~\ref{lem:incr-dim-convex} below says that points from a larger set can be added to a smaller set while fulfilling two seemingly contradictory requirements: adding few enough points to preserve a small convex hull, and adding enough points to generate a large linear span.

\begin{lemma}\label{lem:incr-dim-convex}
For all $S \subseteq E \subseteq \mathbb{R}^d$, if $0 \notin \convh(S)$, there exists $S' \subseteq E$ such that $0 \notin \convh(S \cup S')$ and $\lins(S \cup S') = \lins(E)$.
\begin{proof}
Let $x_1,\dots,x_n \in E \setminus \lins(S)$ be as many linearly independent points as possible, so $\lins(S \cup S') = \lins(E)$, where $S' := \{x_1,\dots,x_n\}$. For all $y_1,\dots, y_k \in S$, if $0 = \sum_{i=1}^n\alpha_ix_i + \sum_{j = 1}^k\beta_jy_j$ is a convex combination, so is $0 = \sum_{j=1}^k\beta_jy_j$ by linear independence. It shows that $0 \notin \convh(S \cup S')$.
\end{proof}
\end{lemma}

Lemma~\ref{lem:incr-dim-convex} above corresponds to, informally, a careful pushing and rotating the hyperplane. It is used in one of the cases in the proof of Lemma~\ref {lem:separation} below.

\begin{lemma}\label{lem:separation}
For all $S \subseteq E \subseteq \mathbb{R}^d$ such that $0 \notin \convh(S)$ and $\rank(E) = d$, there exist affinely independent $x_1,\dots,x_d \in E$ such that $\affh(x_1, \dots, x_d)$ separates $S$ from $0$.
\begin{proof}
If $0 \notin \affh(E)$ then $\dim \affh(E) = d - 1$ by Lemma~\ref{lem:rank-equal} and Observation~\ref{obs:basic-la}, and any $d$ affinely independent points in $E$ witness the claim. Let us assume that $0 \in \affh(E)$, so $\dim \convh(E) = d$ by Lemma~\ref{lem:rank-equal}. By Lemma~\ref{lem:incr-dim-convex} we can assume wlog that $\lins(S) = \mathbb{R}^d$, so $\dim \affh(S) \in \{d, d-1\}$ by Observation~\ref{obs:basic-la}. If $\dim \affh(S) = d$, among the vertices of a well-chosen facet of $\convh(S)$ there are $d$ affinely independent points witnessing the claim. If $\dim \affh(S) = d - 1$ then $0 \notin \affh(S)$ by Lemma~\ref{lem:rank-equal}, so any $d$ affinely independent points from $S$ witness the claim.
\end{proof}
\end{lemma}

\begin{proof}[of Theorem~\ref{thm:minkowsky-single-hyperplane}]
The "if" implication is clear by Observation~\ref{obs:hyp-sep-eq-win}, so let us assume that the game is won by \playerB\/. Let $A_1,\dots,A_n$ be the elements of $\MoveSetA$. By Lemma~\ref{lem:bound}, there exist $a_1\in A_1,\dots,a_n\in A_n$ and $B \in \MoveSetB$ such that $0 \notin \convh(a_1,\dots,a_n) \setsum \convh(B)$, which is equal to $\convh(\{a_1,\dots,a_n\} \setsum B)$. By Lemma~\ref{lem:separation} there exist affinely independent $x_1,\dots,x_d \in (\cup_{i=1}^nA_i \setsum B) \cup \{e_1,\dots,e_d\}$ such that $\affh(x_1, \dots, x_d)$ separates $\{a_1,\dots,a_n\} \setsum B$ from $0$, \textit{i.e.} each $a_i \setsum B$ from $0$.
\end{proof}

Algorithm~\ref{algo:fw} invokes Theorem~\ref{thm:minkowsky-single-hyperplane} to decide the winner of a Minkowski game with finite moves.

\begin{algorithm}
\SetKwProg{Fn}{Function}{ is}{end}
\SetKwFunction{KwFindWinner}{FindWinner}
\SetKwFunction{KwRank}{Rank}
\SetKwFunction{KwDecDim}{DecDim}
\SetKwFunction{KwSep}{Sep}

\Fn{\KwFindWinner}{
\SetKwInOut{Input}{input}\SetKwInOut{Output}{output}
\Input{$d\in \mathbb{N}$, polytopes $A_1,\dots,A_n, B_1\dots,B_m \subseteq \mathbb{Q}^d$}

\Output{the winner of the corresponding Minkowski game}

\BlankLine

Let finite $C$ be the canonical basis of $\mathbb{R}^d$\;

\For(\label{line:j}){$1 \leq j \leq m$}{
	\For(\label{line:d-tuple}){$x_1,\dots,x_d \in (\cup_{i=1}^nA_i \setsum B_j) \cup C$}{
		\If(\label{line:rank}){\KwRank{$x_2-x_1,\dots, x_d - x_1$} $ = d-1$}{
			$w \leftarrow 0$\;
			\For(\label{line:i}){$1 \leq i \leq n$}{
				\For(\label{line:ai}){$a_i \in A_i$}{
				\lIf(\label{line:sep}){\KwSep{$x_1,\dots,x_d, a_i \setsum B_j$}}{$w \leftarrow w + 1$}}
				}
			\lIf{w = n}{\KwRet{"\playerB\/ wins"}
			}
		}
	}	
}
\KwRet{"\playerA\/ wins"}\;	
}

\BlankLine\BlankLine

\Fn{\KwRank}{
\SetKwInOut{Input}{input}\SetKwInOut{Output}{output}
\Input{vectors $x_1,\dots,x_k \in \mathbb{Q}^d$ for some $d\in \mathbb{N}$.}
\Output{the rank of $\{x_1,\dots,x_n\}$}
}

\BlankLine\BlankLine

\Fn{\KwSep}{
\SetKwInOut{Input}{input}\SetKwInOut{Output}{output}
\Input{points $x_1,\dots,x_k \in \mathbb{Q}^d$, finite $Y \subseteq \mathbb{Q}^d$}
\Output{whether the hyperplane including $\{x_1,\dots, x_d\}$ separates $Y$ from $0$}
}
\caption{FindWinner}\label{algo:fw}
\end{algorithm}

\begin{corollary}
\label{corr:fixeddimensionpolytime}
Consider Minkowski games with moves $A_1,\dots,A_n$ for \playerA\/ and $B_1,\dots,B_m$ for \playerB\/ that are finite sets of rational vectors. The algorithmic complexity of deciding the winner of the game is bounded from above by a multivariate polynomial of degree $2d+2$ with leading term $\sum_{i,j}|A_i|^{d+1}|B_j|^{d+1}$.

\begin{proof}
The time required for the rank computation on Line~\ref{line:rank} of Algorithm~\ref{algo:fw} is a function of $d$ so we can ignore it. Given $j$, Line~\ref{line:sep} is reached at most $(\sum_i|A_i||B_j| +d)^d(\sum_i|A_i|)$ times, and the time required to decide separation on Line ~\ref{line:sep} is of the form $f(d)|B_j|$, so the time required by the whole algorithm is of the form $\sum_j|B_j|(\sum_i|A_i||B_j| +d)^d(\sum_i|A_i|)$, which is equivalent to $\sum_{i,j}|A_i|^{d+1}|B_j|^{d+1}$.
\end{proof}
\end{corollary}

\subsection{Moves defined in the first-order theory of the reals}
In this subsection we show that if the moves are definable in the first-order theory of the reals, then so is Condition $8.$ in Lemma~\ref{lem:bound}. As the first-order theory is decidable, so is the question of who is winning a given boundedness Minkowski game with such moves.

We consider first-order formulas with binary function symbols $+$ and $\cdot$, constants $0$ and $1$ and binary relation symbol $<$. A move $A \subseteq \mathbb{R}^d$ is defined by some formula $\phi_A$ with $d$ free variables $x_1, \ldots, x_k$ iff $A = \{(x_1,\ldots,x_k) \in \mathbb{R}^d \mid \phi(x_1,\ldots,x_n)\}$. If $\phi$ defines $A$, then
\begin{align*}\phi_{\textrm{conv}} & = \exists a_1^1, \ldots, a_{d+1}^1, \ldots, a_{d+1}^d, \alpha_1, \ldots, \alpha_{d+1} \quad \bigwedge_{i = 1}^{d+1} \phi(a_i^1,\ldots,a_i^d)\\
& \wedge \sum_{i = 1}^{d+1} \alpha_i = 1 \wedge \bigwedge_{i = 1}^{d+1} 0 \leq \alpha_i \wedge \bigwedge_{j = 1}^d \left ( x_j = \sum_{i=1}^{d+1} \alpha_ia_i^{j} \right )
 \end{align*} defines $\convh(A)$. Also, the formula
 \begin{align*} \phi_{\textrm{cl}} & = \forall \varepsilon \quad \varepsilon > 0 \Rightarrow \left ( \exists a_1, \ldots, a_d \quad \phi(a_1,\ldots,a_d) \wedge \bigwedge_{i = 1}^d a_i < x_i + \varepsilon \wedge x_i < a_i + \varepsilon  \right )\\
 \end{align*}
 defines $\overline{A}$. It then follows that Condition $8.$ in Lemma~\ref{lem:bound} is expressible as some formula $\phi_{\textrm{win}}$ obtained from the formulas $\phi_A$, $\phi_B$ defining the moves in $\MoveSetA$ and $\MoveSetB$. Moreover, the length of the formula $\phi_{\textrm{win}}$ is polynomially bounded in the sum of the length of the $\phi_A$, $\phi_B$.

\begin{proposition}
\label{prop:firstorder}
Deciding the winner of a boundedness Minkowski game with moves defined in the first-order theory of the reals is {\sc 2EXPTIME-complete}.
\begin{proof}
As explained above, deciding whether \playerA\/ wins reduces to deciding whether $\phi_{\textrm{win}}$ is true, and by \cite{ben-or}, this can be done in {\sc 2EXPTIME}. For hardness, note that deciding truth of a formula $\phi$ is {\sc 2EXPTIME-hard} by \cite{davenport}. This reduces to our problem by considering the one-dimensional one-sided Minkowski game $\langle 0, \{A\}, \{0\}\rangle$ where $A = \{x \mid \left (x = 0 \wedge \phi \right ) \vee \left (x = 1 \wedge \neg \phi \right )\}$. Clearly, \playerA\/ wins $\langle 0, \{A\}, \{0\}\rangle$ iff $\phi$ is true.
\end{proof}
\end{proposition}

\subsection{The computable analysis perspective}
\label{subsec:caboundedness}
If we represent the sets involved in the boundedness Minkowski games via polyhedra or first order formula, we have only restricted expressivity available to us. Using notions from computable analysis \cite{weihrauchd}, we can however consider computability for all boundedness Minkowski games with closed moves -- and as Lemma~\ref{lem:bound} demonstrated, this is not a problematic restriction. As the involved spaces are all connected, we cannot expect decidability, and instead turn our attention to semidecidability, i.e.~truth values in the Sierpinski space $\mathbb{S}$.

We do have to decide on a representation for the sets, though. Pointing to \cite{pauly-synthetic} for definitions and explanations, we have the spaces $\mathcal{A}(\mathbb{R}^d)$ of closed subsets, $\mathcal{K}(\mathbb{R}^d)$ of compact subsets and $\mathcal{V}(\mathbb{R}^d)$ of overt subsets available. In $\mathcal{A}(\mathbb{R}^d)$, a closed subset can be seen as being represented by an enumeration of rational balls exhausting its complement. The space $\mathcal{K}(\mathbb{R}^d)$ adds to that some $K \in \mathbb{N}$ such that the set is contained in $[-K,K]^d$. In $\mathcal{V}(\mathbb{R}^d)$, a closed subset is instead represented by listing all rational balls intersecting it.

A relevant property is that universal quantification over compact sets from $\mathcal{K}(\mathbb{R}^d)$ and existential quantification over overt sets from $\mathcal{V}(\mathbb{R}^d)$ preserve open predicates. We can use the former to find that:

\begin{proposition}
The Minkowski sum $\mathalpha{+} : \mathcal{A}(\mathbb{R}^d) + \mathcal{K}(\mathbb{R}^d) \to \mathcal{A}(\mathbb{R}^d)$ is computable.
\begin{proof}
$\mathalpha{\notin} : \mathbb{R}^d \times \mathcal{A}(\mathbb{R}^d) \to \mathbb{S}$ is an open predicate by definition. Now note that $y \notin A + B$ iff $\forall z \in B \ y - z \notin A$.
\end{proof}
\end{proposition}

The Minkowski sum of two closed sets is not computable as a closed set: consider some $A \in \mathcal{A}(\mathbb{N})$. Then $0 \in A + (-A)$ iff $A \neq \emptyset$. If $+$ were computable, the former would be a $\Pi^0_1$-property, whereas the latter is $\Pi^0_2$-complete, and we find a contradiction.

It was already shown in \cite[Proposition 1.5]{paulyleroux} (also \cite{ziegler8}) that convex hull is a computable operation on compact sets, but not on closed sets. Put together, we find that:

\begin{proposition}
\label{prop:computableanalysis}
Consider boundedness Minkowski games, where moves $A \in \MoveSetA$ are given as overt sets (i.e.~in $\mathcal{V}(\mathbb{R}^d)$) and moves $B \in \MoveSetB$ are given as compact sets (i.e.~in $\mathcal{K}(\mathbb{R}^d$)). The set of games won by \playerB\/ constitutes a computable open subset.
\begin{proof}
By Lemma~\ref{lem:bound}, \playerB\/ can win iff for $\MoveSetA = \{A_0, \ldots, A_n\}$ we find that there exists $a_i \in A_i$ and $B \in \MoveSetB$
 \[0 \notin (\convh \{a_i \mid i \leq n\}) + \convh(\overline{B})\]

 As $B$ is given as a compact set, we can compute $(\convh \{a_i \mid i \leq n\}) + \convh(\overline{B}) \in \mathcal{A}(\mathbb{R}^d)$. As before, $\mathalpha{\notin} : \mathbb{R}^d \times \mathcal{A}(\mathbb{R}^d) \to \mathbb{S}$ is an open predicate, and existential quantification over overt sets preserves open predicates. Thus, the entire requirements define a computably open subset of the space of Minkowski games.
\end{proof}
\end{proposition}

\section{The winning region of the safety problem}\label{wrsp}
We now turn our attention to safety Minkowski games. Given some move sets $\MoveSetA$, $\MoveSetB$ and the safe zone $\SafeZone$, we want to understand for which initial positions $\init \in \SafeZone$ \playerA\/ has a winning strategy in the safety Minkowski game $\langle \MoveSetA, \MoveSetB, \SafeZone\rangle$. In a minor abuse of notation, we speak of {\bf the} safety Minkowski game $\langle \MoveSetA, \MoveSetB, \SafeZone\rangle$, and call the set of $\init$ such that \playerA\/ has a winning strategy the \emph{winning region} $W$.

Let us first note that these games are determined by Borel determinacy~\cite{martin} (and also from techniques in \cite{GS53}) since the set of the winning plays for \playerA\/ is a closed set (for the usual product topology over discrete topology).

We give three kinds of general results: first, the winning region is the greatest fixed point of an operator that removes the points where \playerB\/ can provably win (in finitely many rounds); second, topological and finiteness assumptions about the game implies topological properties of the winning region and of its boundary; third, the winning region is \emph{stuck} inside \SafeZone, \textit{i.e.} for all non-zero translation, there is a translation in the same direction, with smaller module, and not sending $W$ fully inside $\SafeZone$.

Let $\langle \MoveSetA, \MoveSetB, \SafeZone \rangle$ be a safety game. Given $E$ a target set, $f(E)$ is defined below as the positions from where \playerA\/ can ensure to fall in $E$ after one round of the game.

\begin{definition}
\label{def:fixedpointfunction}
For all $E \subseteq \mathbb{R}^d$ let
\[f(E) := \{x\in \mathbb{R}^d\,\mid\, \exists A\in \MoveSetA, \forall a\in A, \forall B\in \MoveSetB, \exists b\in B, x + a + b \in E\}\]
and let $g(E) := f(E) \cap \SafeZone$.
\end{definition}

Note that the fixed-point characterization of the winning region requires no assumption.

\begin{lemma}\label{lem:winning-fixed-point}
The winning region $W$ of \playerA\/ is the greatest fixed point of $g$, even for infinite $\MoveSetA$ and $\MoveSetB$.

\begin{proof}
First note that every fixed point of $g$ is included in $W$, since starting from there allows \playerA\/ to stay there for one round, and therefore forever. Therefore it suffices to show that $g(W) = W$, which holds since being in $W$ is equivalent to being in \SafeZone\/ and able to reach $W$ in one round.
\end{proof}
\end{lemma}

At the cost of a finiteness assumption on $\MoveSetA$ below, we invoke the Kleene fixed point theorem and show that the winning region can be computed in $\omega$ many steps.

\begin{proposition}\label{prop:omegaiteration}
Let $S_0 :=  \mathbb{R}^d$, let $S_{n+1} := g(S_n)$ for all $n$, and let $S_{\omega} := \cap_{n\in \mathbb{N}}S_n$. $S_{\omega}$ is the greatest fixed point of $g$, even for infinite $\MoveSetB$.

\begin{proof}
First note that $f$, then $g$ are monotone. To prove the lemma it suffices to invoke (the dual of) the Kleene fixed point theorem, after proving (the dual of) the Scott continuity, namely that if $(E_n)_{n\in\mathbb{N}}$ satisfy $E_{n+1} \subseteq E_n \subseteq \mathbb{R}^d$ for all $n$, then $g(\cap_nE_n) = \cap_n g(E_n)$. We prove it for $f$ below, then it holds clearly for $g$, too.

$f(\cap_nE_n) \subseteq \cap_n f(E_n)$ by monotonicity. Conversely, let $x \in \cap_n f(E_n)$, so for all $n$ there is $A_n \in \MoveSetA$ such that $x + a \setsum B \cap E_n \neq \emptyset$ for all $a\in A_n$ and $B\in \MoveSetB$. By finiteness of  $\MoveSetA$ there is a constant subsequence $(A)$ of $(A_n)$, defined by some $\varphi$, so $x + a \setsum B \cap E_{\varphi(n)} \neq \emptyset$ for all $a\in A$ and $B\in \MoveSetB$. So $x \in f(\cap_nE_n)$, since $\cap_nE_{\varphi(n)} = \cap_nE_n$.
\end{proof}
\end{proposition}

The example below shows that the restriction to games with finitely many moves for $\MoveSetA$ is necessary in Proposition \ref{prop:omegaiteration}. The game is essentially a single player game, i.e.~\playerB\/ has no non-trivial choices to make.
\begin{example}
Consider the $5$-dimensional safety Minkowski game $\langle \MoveSetA_{\textrm{init}} \cup \MoveSetA_{\textrm{up}} \cup \MoveSetA_{\textrm{down}}, \{\{0\}\}, \SafeZone\rangle$, where:
\begin{itemize}
\item $\SafeZone = \uint^2 \times \convh(\{(x,x,1) \in \uint^3 \} \cup \uint \times \{(0,0)\})$.
\item $\MoveSetA_{\textrm{init}} = \{(0,+1,+2^{-k},0,0) \mid k \in \mathbb{N}\}$
\item $\MoveSetA_{\textrm{up}} = \{(+2^{-k}, 0, 0, +2^{-k}, +1) \mid k \in \mathbb{N}\}$
\item $\MoveSetA_{\textrm{down}} = \{(+2^{-k}, 0, 0, -2^{-k}, -1) \mid k \in \mathbb{N}\}$
\end{itemize}
Then $S_{\omega} = [0,1) \times \{(0,0,0,0\}$, but $W = \emptyset$.
\begin{proof}
Every finite non-losing play has to start with some move from $\MoveSetA_{\textrm{init}}$, which sets the third component to some particular $2^{-k_0}$. Afterwards, moves from $\MoveSetA_{\textrm{up}}$ and $\MoveSetA_{\textrm{down}}$ for the same value $k = k_0$ have to alternate, otherwise the safety constraint is immediately violated. As any such move increased the first component by $2^{-k_0}$, we find that after some finite number of moves the first component exceeds $1$, and the game is lost. However, this number can be chosen arbitrarily large based on the choice of $k_0$.
\end{proof}
\end{example}

As a tangential remark, note that the example above could be adapted to a game with only finitely many moves, but cooperative players by joining all moves of the same type into one, and letting \playerB\/ choose $k$. Thus, considering non-zero sum games would also change the fixed point iteration.

Below, compactness of the winning region follows from topological and finiteness assumptions. Other assumptions (note the symmetry) make the interior of $S_{\omega}$ unreachable by \playerA\/ from its boundary. If $W = S_{\omega}$, this gives its boundary the status of "almost-tie" region.

\begin{proposition}
\begin{enumerate}
\item Let the elements of (possibly infinite) $\MoveSetB$ be closed. If $\SafeZone$ is compact, so are the $S_n$ and $S_{\omega}$.

\item Let the elements of (possibly infinite) $\MoveSetA$ be compact. In every game $\langle \MoveSetA, \MoveSetB, \SafeZone, \init \rangle$ with $\init \in S_{\omega} \setminus \mathring{S}_{\omega}$, \playerA\/ cannot force any end-of-the-round position inside $\mathring{S}_{\omega}$.
\end{enumerate}

\begin{proof}
\begin{enumerate}
\item Let $E$ be a closed set and let $x \notin f(E)$. For all $A\in \MoveSetA$ let $a_A$ and $B_A$ be such that $x + a_A \setsum B_A \cap E = \emptyset$. Since $B_A$ and $E$ are closed, there is $r_A > 0$ such that $B(x,r_A) \setsum a_A \setsum B_A \cap E = \emptyset$. Let $r := \min_A r_A$. For all $y \in B(x,r)$ for all $A \in \MoveSetA$ we find $y + a_A \setsum B_A \cap E = \emptyset$. It shows that $\mathbb{R}^d \setminus f(E)$ is open, so $f$ preserves closeness, and so does $g$. Therefore $S_n$ is compact (closed and bounded) for all $n$, and so is $S_{\omega}$ by intersection.

\item It suffices to show that $\forall A \in \MoveSetA \exists a \in A \exists B \in \MoveSetB \forall b \in B, x + a + b \notin \mathring{S}_{\omega}$, so let $A \in \MoveSetA$. Towards a contradiction, let us assume that $\forall a \in A, \forall B \in \MoveSetB \exists b \in B, x + a + b \in \mathring{S}_{\omega}$. Let us fix $B \in \MoveSetB$ for now, so $\forall a \in A, \exists b_a \in B, x + a + b_a \in \mathring{S}_{\omega}$. Since $\mathring{S}_{\omega}$ is open, $\forall a \in A, \exists b_a \in B \exists r_a > 0, x + a + b_a \setsum B(0,2r_a) \subseteq \mathring{S}_{\omega}$. The open balls $\{B(a,r_a)\}_{a \in A}$ form a cover of $A$, so by compactness let finitely many $a_1,\dots, a_k$ be such that the $B(a_i,r_{a_i})$ still cover $A$. For all $a \in A$ we can thus define $i(a) := \min \{i \,\mid\, a \in B(a_i,r_{a_i})\}$ and $b'_a := b_{a_{i(a)}}$. Let $r := \min_i \{r_{a_i}\}$. For all $a \in A$ and $\delta \in \mathbb{R}^d$ such that $\|\delta\| < r$, we have $x + \delta + a + b'_a = x + \delta + a_{i(a)} + (a - a_{i(a)}) + b_{a_{i(a)}}$. Since $\|a - a_{i(a)}\| < r_{a_{i(a)}}$ by definition of the cover and $\|\delta\| < r \leq r_{a_{i(a)}}$, we find  $x + \delta + a + b'_a \in \mathring{S}_{\omega}$. So $\forall y \in B(x,r)\forall a \in A \exists b \in B, y + a + b \in  \mathring{S}_{\omega}$. Just before letting $B$ range over $\MoveSetB$ again, let $r_B := r$. Since $\MoveSetB$ is finite, let $r' := \min_{B \in \MoveSetB} \{r_B\}$. Therefore $\forall y \in B(x,r')\forall a \in A \forall B \in \MoveSetB \exists b \in B, y + a + b \in \mathring{S}_{\omega}$, so $x \in  \mathring{S}_{\omega}$, contradiction.
\end{enumerate}
\end{proof}
\end{proposition}

Finally, we give geometrical properties of the winning region wrt \SafeZone. Lemma~\ref{lem:gen-set-stuck}, requiring no assumption, says the following: seen as a physical object, $W$ cannot move by (continuous) translation to another position while always remaining entirely in \SafeZone.

\begin{lemma}\label{lem:gen-set-stuck}
For all $t\in\mathbb{R}^d \setminus\{0\}$, either $\mathbb{R}^+\cdot t \setsum W \subseteq W$, or for all $\epsilon > 0$ there is $0 < \epsilon' \leq \epsilon$ such that $(\epsilon'\cdot t \setsum W) \cap \SafeZone^C \neq \emptyset$.

\begin{proof}
The proof has two parts. First, let us prove the following claim about general sets: if $S \subseteq T \subseteq \mathbb{R}^d$ are such that $\forall t \in \mathbb{R}^d(t \setsum S \subseteq T \Rightarrow t \setsum S \subseteq S)$, then for all $t\in\mathbb{R}^d \setminus\{0\}$, either $\mathbb{R}^+\cdot t \setsum S \subseteq S$, or for all $\epsilon > 0$ there is $0 < \epsilon' \leq \epsilon$ such that $(\epsilon'\cdot t \setsum S) \cap T^C \neq \emptyset$. Let $t\in \mathbb{R}^d\setminus\{0\}$ and let $x + \epsilon_0\cdot t\notin S$ for some $\epsilon_0 > 0$ and $x \in S$. Towards a contradiction, let $\epsilon_1 > 0$ be such that $(S + \epsilon\cdot t) \subseteq T$ for all $ 0 < \epsilon \leq \epsilon_1$. Let $n\in\mathbb{N}$ be such that $\frac{\epsilon_0}{n} \leq \epsilon_1$. On the one hand $\frac{\epsilon_0}{n} \cdot t + S \subseteq S$, so $\frac{k\epsilon_0}{n} \cdot t + S \subseteq S$ for all $k\in\mathbb{N}$, by induction on $k$. On the other hand there exists a natural $0 < k \leq n$ such that $\neg (\frac{k\epsilon_0}{n} \cdot t + S \subseteq S)$, contradiction, and the claim is proved.

Let us now prove that for all $t\in\mathbb{R}^d$, if $(t + W) \subseteq \SafeZone$ then $t + W \subseteq W$.
If $(t + W) \subseteq \SafeZone$, \playerA\/ can stay in $\SafeZone$ when starting in $t + W$, simply by using a winning strategy for $W$ up to translation by $t$. So $t + W \subseteq W$. Invoking the above claim shows the lemma.
\end{proof}
\end{lemma}

The following is a corollary of  Lemma~\ref{lem:gen-set-stuck}: it says that if $\SafeZone$ is bounded and convex, the image of $W$ by any non-zero translation is no longer included in $\SafeZone$.

\begin{corollary}
If $\SafeZone$ is bounded and convex, $(t + W) \cap \SafeZone^C \neq \emptyset$ for all $t\in\mathbb{R}^d \setminus\{0\}$.
\begin{proof}
Towards a contradiction, let $t\in\mathbb{R}^d \setminus\{0\}$ be such that $(t + W) \cap \SafeZone^C = \emptyset$. By Lemma~\ref{lem:gen-set-stuck} let $0< \epsilon < 1$ be such that $(\epsilon\cdot t + W) \cap \SafeZone^C \neq \emptyset$, which is witnessed by some $x \in W$ such that $\epsilon\cdot t + x \notin \SafeZone$. By convexity $t + x   \notin \SafeZone$, contradiction.
\end{proof}
\end{corollary}

\hide{Interesting example for something: let $C := \{(\cos \theta,\sin\theta,0) \mid\theta \in[0,2\pi[\}$ and let $\SafeZone := \convh(C \cup \{(\cos \theta,\sin\theta,\theta)\mid \theta \in[0,2\pi[\})$.}

\section{Computational complexity of the safety problems}
\label{sec:undecidable}

We start our investigation of the computational complexity of determining the winner in safety Minkowski games by considering the general setting of computable analysis, as we did in Subsection \ref{subsec:caboundedness} for the boundedness games. We point again to \cite{pauly-synthetic} for notation and definition, and in particular make use of the characterizations of $\mathcal{V}(\mathbb{R}^d)$ and $\mathcal{K}(\mathbb{R}^d)$ via the preservation of open predicates under quantification. We obtain:

\begin{observation}
\label{obs:casafety}
Consider finite $\MoveSetA$ of moves from $\mathcal{V}(\mathbb{R}^d)$ and finite $\MoveSetB$ of moves from $\mathcal{K}(\mathbb{R}^d)$. Let $\SafeZone \in \mathcal{A}(\mathbb{R}^d)$. Then the function $g$ from Definition \ref{def:fixedpointfunction} is well-defined and computable from the parameters as a function $g : \mathcal{A}(\mathbb{R}^d) \to \mathcal{A}(\mathbb{R}^d)$.
\end{observation}

\begin{proposition}
\label{prop:casafety}
Given a safety Minkowski game $\langle \MoveSetA, \MoveSetB, \SafeZone, \init\rangle$ with $\MoveSetA$ being a finite set of overt sets (i.e.~from $\mathcal{V}(\mathbb{R}^d)$), $\MoveSetB$ being a finite set of compact sets (i.e.~from $\mathcal{K}(\mathbb{R}^d)$) and $\SafeZone$ being given as an element of $\mathcal{A}(\mathbb{R}^d)$, we can semidecide (recognize) if \playerB\/ has a winning strategy.
\begin{proof}
By Observation \ref{obs:casafety}, we can compute the function $g : \mathcal{A}(\mathbb{R}^d) \to \mathcal{A}(\mathbb{R}^d)$ defined in Definition \ref{def:fixedpointfunction}. As $\mathcal{A}(\mathbb{R}^d)$ is effectively closed under countable intersection, we then can compute $S_\omega \in \mathcal{A}(\mathbb{R}^d)$. By Proposition \ref{prop:omegaiteration}, this is the greatest fixed point of $g$, and by Lemma \ref{lem:winning-fixed-point}, the greatest fixed point is the winning region of \playerA\/. Thus, \playerB\/ wins iff $\init \notin S_\omega$, and by definition of $\mathcal{A}(\mathbb{R}^d)$, this is semidecidable.
\end{proof}
\end{proposition}

In the remainder of this section, we prove that safety Minkowski games are undecidable even if moves and safe zone are defined as a set of rational linear inequalities.

\begin{theorem}
\label{thm:complexitysafetypolytope}
There is $d \in \mathbb{N}$, a rational convex polytope $\SafeZone$ and a finite family $\MoveSetA$ of rational closed convex polytopes all in $\mathbb{R}^d$ such that it is undecidable, whether \playerA\/ has a winning strategy in the one-sided safety Minkowski game $\langle \MoveSetA, \SafeZone, \init \rangle$, given $\init$ as a rational vector.
\end{theorem}

To establish this theorem, we provide a reduction from the control state reachability problem for two counter machines to the problem of deciding if \playerB\/ has a winning strategy in a safety Minkowski game. As the first step, we introduce a slightly more general version of one-sided Minkowski games, and demonstrated a reduction to one-sided safety Minkowski games:

\begin{definition}
A safety-reachability one-sided Minkowski game is given by a tuple $\langle \MoveSetA, \SafeZone, \GoalZone, \init\rangle$, where $\langle \MoveSetA, \SafeZone, \init\rangle$ is some $d$-dimensional safety one-sided Minkowski game, and $\GoalZone \subseteq \SafeZone$. It is played like the safety Minkowski game, and if \playerA\/ wins $\langle \MoveSetA, \SafeZone, \init\rangle$, then he wins $\langle \MoveSetA, \SafeZone, \GoalZone, \init\rangle$. If the play enters $\GoalZone$ prior to leaving $\SafeZone$ for the first time, also \playerA\/ wins. Else \playerB\/ wins.
\end{definition}

\begin{proposition}
Given a $d$-dimensional safety-reachability one-sided Minkowski game $\langle \MoveSetA, \SafeZone, \GoalZone, \init\rangle$, we define the associated $d+1$-dimensional safety one-sided Minkowski game $\langle \MoveSetA', \SafeZone',\init'\rangle$ as follows:
\begin{enumerate}
\item $\init' := \langle \init, 0\rangle$
\item $\SafeZone' := \convh \left ( (\SafeZone \times \{0\}) \cup (\GoalZone \times \{1\})\right )$
\item $\MoveSetA' := \{A \times \{0\} \mid A \in \MoveSetA\} \cup \{\{(0,\ldots,0,1)\}, \{(0,\ldots,0,-1)\}\}$
\end{enumerate}
Now \playerA\/ (resp.~\playerB\/) has a winning strategy in the original game iff he (resp.~she) has one in the associated game.
\begin{proof}
Every play in $\langle \MoveSetA', \SafeZone',\init'\rangle$ where \playerA\/ never chooses one of the moves in $\{\{(0,\ldots,0,1)\}, \{(0,\ldots,0,-1)\}\}$ is also a valid play in $\langle \MoveSetA, \SafeZone, \GoalZone, \init\rangle$ after projection, and \playerB\/ has no additional options to deviate in the latter. Moreover, if the play is won for \playerA\/ in $\langle \MoveSetA', \SafeZone',\init'\rangle$, then it is also winning for her in $\langle \MoveSetA, \SafeZone, \GoalZone, \init\rangle$.

By construction of $\SafeZone'$, the first time \playerA\/ uses a move from $\{\{(0,\ldots,0,1)\}, \{(0,\ldots,0,-1)\}\}$, it has to be $\{(0,\ldots,0,1)\}$, and the first $d$ components of the position need to fall into $\GoalZone$. Thus, in the corresponding play in $\langle \MoveSetA, \SafeZone, \GoalZone, \init\rangle$, \playerA\/ has already won. Conversely, if the play reaches $\GoalZone$ in $\langle \MoveSetA, \SafeZone, \GoalZone, \init\rangle$, then \playerA\/ can continue the play in $\langle \MoveSetA', \SafeZone',\init'\rangle$ by alternating the moves $\{(0,\ldots,0,1)\}$ and $\{(0,\ldots,0,-1)\}$ and win.
\end{proof}
\end{proposition}

We recall some preliminaries on two-counter machines:.

\paragraph{{\bf 2CM and the control state reachability problem}}
A two-counter machine, {\sf 2CM} for short, is defined by a finite directed graph $(Q,E)$ with labeled edges. Vertices have out-degree $0$, $1$ or $2$. If the out-degree is $1$, the corresponding edge is labeled with one of $\mathrm{INC}^i$, $\mathrm{DEC}^i$ for $i \in \{0,1\}$. If the out-degree is $2$, one outgoing edge is labeled with $\textrm{isZero?}^i$ and the other with $\textrm{isNotZero?}^i$ for some $i \in \{0,1\}$. There is a designated starting vertex $q_0 \in Q$.

A finite or infinite path through the graph is a \emph{valid computation starting from} $n_0$ and $n_1$ if the following is true: the path starts at $q_0$. If one starts with $c_0 := n_0$ and $c_1 := n_1$ and increments (decrements) $c_i$ by $1$ whenever encountering a label $\mathrm{INC}^i$ ($\mathrm{DEC}^i$), then at the moment an edge labeled with $\textrm{isZero?}^i$ ($\textrm{isNotZero?}^i$) is passed, we find that $c_i = 0$ ($c_i \neq 0$). Moreover, we demand that a decrement command is never encountered for a counter with value $0$.

\begin{theorem}[{\cite[Theorem Ia]{minsky}}]
There is a {\sf 2CM} such that it is undecidable whether there exists an infinite valid computation starting from $n_0$ and $n_1$ (where $n_0$ and $n_1$ are the input).
\end{theorem}

We will slightly modify the {\sl 2CM} to simplify the construction. We subdivide every edge by adding another vertex on it. If the original edge was labeled $\textrm{INC}^i$ ($\textrm{DEC}^i$), then the two new edges will be labeled $\textrm{INCa}^i$ and $\textrm{INCb}^i$ ($\textrm{DECa}^i$ and $\textrm{DECb}^i$). If the original edge was labeled $\textrm{isZero?}^i$ or $\textrm{isNotZero?}^i$, we move the label to the newly-added vertex.

Now we are ready to reduce the non-halting problem of modified {\sf 2CM}'s to the existence of a winning strategy for \playerA\/ in a safety-reachability one-sided Minkowski game. The general idea of the reduction is as follows. First, \playerA\/ is forced to simulate the computation of the 2CM in order to avoid violating the safety condition of the safety Minkowski game. The value of each counter $c_i$, $i \in \{1,2\}$, is coded in some dimension $y_i$ such that when the counter $c_i$ is equal to $k \in \mathbb{N}$ then the value of $y_i = \frac{1}{2^k}$. The role of \playerB\/ is restricted to assist \playerA\/ to multiply or divide the $x_i$ by $2$. Her failure to operate as intended will let the play reach $\GoalZone$. Additionally, each vertex $Q$ is associated with one dimension that will be non-zero iff the computation is currently in that vertex.

All the moves and invariants that we use are definable by finite sets of linear constraints.

\paragraph*{Defining the reduction}
We are given a modified {\sf 2CM} with vertex set $Q$ (called control states) and edges $E$. The associated safety-reachability Minkowski game will be played in $\mathbb{R}^{4 + |Q|}$. The first $4$ dimensions are $(x_0,y_0,x_1,y_1)$, where the $y_i$ encode the counter values, and the $x_i$ are auxiliary values. The remaining $|Q|$ dimensions shall be indexed with the states $q$.

Every instruction $e \in E$ corresponds to some move $A_e$ for \playerA. The move $A_e$ will always decompose as $A_e = A_e^{xy} \times \{a_e^Q\}$. If $e$ is an edge from $q_i$ to $q_f$, then $a_e^Q \in \mathbb{R}^{|Q|}$ will have $-1$ at component $q_i$, $+1$ at component $q_f$ and $0$ elsewhere.

\begin{tabular}{c|c|c|c}
Label of $e$ & Value of $A_e^{xy}$ & Label of $e$ & Value of $A_e^{xy}$\\
\hline
- & $\{(0,0,0,0)\}$ \\
$\textrm{INCa}^0$ & $\convh \{(0,0),(1,-1)\} \times \{(0,0)\}$	&	$\textrm{DECa}^0$ & $\convh \{(0,0),(1,0)\} \times \{(0,0)\}$\\
$\textrm{INCa}^1$ & $\{(0,0)\} \times \convh \{(0,0),(1,-1)\}$	&	$\textrm{DECa}^1$ & $\{(0,0)\} \times \convh \{(0,0),(1,0)\}$\\
$\textrm{INCb}^0$ & $\convh \{(0,0), (-1,0)\} \times \{0,0\}$	&	$\textrm{DECb}^0$ & $\convh \{(0,0),(-1,1)\} \times \{(0,0)\}$\\
$\textrm{INCb}^1$ & $\{0,0\} \times \convh \{(0,0), (-1,0)\}$	&	$\textrm{DECb}^1$ & $\{(0,0)\} \times \convh \{(0,0),(-1,1)\}$\\
\end{tabular}

It remains to define the sets $\SafeZone$ and $\GoalZone$. For that, let $Q_z^i$ be the set of states labeled with $\textrm{isZero?}^i$, and let $Q_n^i$ be the set of states labeled with $\textrm{isNotZero?}^i$. Let $Q_o$ be the set of unlabeled states with non-zero outdegree. Let $e_q$ be the $|Q|$-dimensional vector having $1$ in component $q$ and $0$ elsewhere.

\begin{align*}
\SafeZone := \convh  [ & \left ( \bigcup_{q \in Q_o} [0,1]^4 \times \{e_q\} \right ) \cup \left ( \bigcup_{q \in Q_n^0} [0,1] \times [0,0.7] \times [0,1]^2 \times \{e_q\} \right ) \\
\cup &  \left ( \bigcup_{q \in Q_n^1} [0,1]^3 \times [0,0.7] \times \{e_q\} \right )  \cup  \left ( \bigcup_{q \in Q_z^0} [0,1] \times \{1\} \times [0,1]^2 \times \{e_q\} \right )\\
\cup  & \left ( \bigcup_{q \in Q_z^1} [0,1]^3 \times \{1\} \times \{e_q\} \right ) ]\\
\end{align*}

\[\GoalZone := \SafeZone \cap (\{(x,y) \in \mathbb{R}^2 \mid y \neq x \neq 0\} \times \mathbb{R}^{2+|Q|} \cup \mathbb{R}^2 \times \{(x,y) \in \mathbb{R}^2 \mid y \neq x \neq 0\} \times \mathbb{R}^{|Q|})\]

The starting position of the game is as follows: $(0,2^{-n_0},0,2^{-n_1},0,\ldots,0,1,0,\ldots,)$, where $n_0$ and $n_1$ are the starting values for the counters, and the unique $1$ in the latter part is found at the index corresponding to the starting state of the {\sl 2CM}.

\paragraph{Correctness of the reduction}
We claim that \playerA\/ has a winning strategy in the constructed game, iff the (modified) {\sl 2CM} has a valid infinite computation path. As moves correspond to edges, every sequence of moves chosen by \playerA\/ in the game can be seen as a sequence of edges for the {\sl 2CM}.

First we argue that every sequence of edges which is not a path induces a losing strategy in the game. As the values of the components associated with the control states must remain between $0$ and $1$, and every move has components $-1$, $+1$ somewhere and $0$ elsewhere it follows that every non-losing sequences of moves ensure that exactly one state-component $q_i$ of the position is $1$, and the others are $0$. Every move coming from an edge not having the initial state $q_i$ will lose immediately.

Next, we shall explain how the moves for $\textrm{INCa}^i$ and $\textrm{INCb}^i$ together cause the desired effect. If the current relevant part of the position is $(0,2^{-k})$, then after the move $\textrm{INCa}^i$ \playerB\/ may pick any $(x,y) \in (0,2^{-k}) + \convh \{(0,0),(1,-1)\}$, in other words, \playerB\/ picks some $t \in [0,1]$ and sets the position to $(t,2^{-k} - t)$. If \playerB\/ picks $t = 0$, then \playerA\/ can repeat the same move. By the definition of $\GoalZone$, the only other safe choice for \playerB\/ is to pick $t = 2^{-k-1}$, i.e.~to set the position to $(2^{-k-1},2^{-k-1})$. The move associated with $\textrm{INCb}^i$ follows, which means that \playerB\/ gets to pick some $(2^{-k-1}-t,2^{-k-1})$. Again, choosing $t = 0$ lets \playerA\/ repeat her move, and the only other choice compatible with avoiding \GoalZone is to move to $(0,2^{-k-1})$.

The construction for $\textrm{DECa}^i$, $\textrm{DECb}^i$ works similarly: starting at $(0,2^{-k})$ for $k \neq 0$, \playerB can only remain, enter $\GoalZone$ or move to $(2^{-k},2^{-k})$ if \playerA\/ plays a move corresponding to $\textrm{DECa}^i$. The subsequent $\textrm{DECb}^i$ move allows \playerB\/ to remain, enter $\GoalZone$ or to move to $(0,2^{-k+1})$. If a $\textrm{DECa}^i$, $\textrm{DECb}^i$-pair is encountered starting at $(0,1)$, \playerB\/ can force the play to leave $\Safe$, corresponding to our convention that decrementing a counter at value $0$ terminates the computation of the {\sf 2CM}.

Finally, we need to discuss (conditional) halting: by the construction of \SafeZone, if a vertex with out-degree $0$ is reached, or a vertex labeled with an unsatisfied condition, then the play is losing for \playerA. Thus, winning strategies of \playerA\/ correspond exactly to infinite non-halting computations of the {\sf 2CM}.

\begin{remark}[On the existence of a finite bisimilarity quotient]
\label{rem:finitebisi}
In line with the undecidability result above, it can be shown that safety Minkowski games with safety region and moves defined by linear inequalities have in general no finite bisimilarity quotient. In contrast, it is an easy exercise to establish, by application of definitions and results in~\cite{HenzingerK97a}, that every safety Minkowski game with a safety region and moves defined as finite union of rational multi-rectangles has a finite bisimilarity quotient.  This finite bisimilarity quotient can then be used to show that the fixed point defining the set of winning states for \playerA\/ is effectively computable. Rational multi-rectangular sets in $\mathbb{R}^d$ are defined as finite union of sets defined by constraints of the form $\bigwedge_{i=1}^{i=d} x_i \in [a_i,b_i]$ where $a_i,b_i \in \mathbb{Q}$ are the rational bounds of an closed non-empty interval in $\mathbb{R}$.
\end{remark}

\section{Structural safety games}\label{sec:ssg}

The undecidability result of the previous section for safety game with polytopic sets motivates the study of {\em structural safety Minkowski games}. In a (one-sided) structural safety game, there is no designated initial state and \playerA\/ is asked to be able to maintain the system safe starting from any point in the safe region. It is not difficult to see that this stronger requirement makes the game equivalent to a "one round" game. Indeed, if \playerA\/ can maintain safety from all positions within $\SafeZone$, then it means that after one round of the game, the game is again within $\SafeZone$, from which \playerA\/ can win for one more round, etc.

We establish in this section the exact complexity of the structural safety games when moves and the set $\Safe$ are polytopic.

\begin{theorem}
\label{theo:structural}
Given a one-sided structural safety Minkowski game $\langle \MoveSetA, \MoveSetB, \SafeZone \rangle$ where moves and the set $\Safe$ are rational polytopic, it is {\sc coNP-Complete} to decide if \playerA\/ has a winning strategy from all positions in $\SafeZone$.
\end{theorem}

To prove this theorem, we first show that when \playerB\/ wins the structural safety Minkowski game then there exists a position $v_0$ and vertex strategy that is winning. This establishes membership of the decision problem to {\sc coNP}.

\begin{lemma}
\label{lem:simple}
Given a one-sided structural safety Minkowski game $\langle \MoveSetA, \SafeZone \rangle$ where moves in $\MoveSetA$ and the set $\Safe$ are rational polytopic, if there is no winning strategy for \playerA\/ then there exists a rational position $v \in \SafeZone$ with polynomial size binary representation and a vertex strategy of \playerB\/ that is winning for \playerB.
\end{lemma}
\begin{proof}
As there is no winning strategy for \playerA\/ in the structural safety game then, by definition, there exists $v \in \Safe$ such that for all $A \in \MoveSetA$, $v \setsum A \not\subseteq \Safe$. As $\Safe$ is convex, it must be the case that for each $A$, there exists $a \in {\sf Ver}(A)$ such that $v+a \not\in \Safe$. Let us note $V$ this set of vertices. As all vertices $a \in V$ are such that $v+a$ is outside $\Safe$, it must be the case that $v+a$ violates at least one of the linear constraints that define $\Safe$. Let $\phi_{a}$ be one such constraint defining $\Safe$ and violated by $v+a$ for move $A$. So we can deduce that the following system of inequalities is satisfiable: $x \in \Safe \land \bigwedge_{a \in V} x+a \not\in \sem{\phi_ a}$. Because all the vertices $a$ are definable by polynomial size binary representations and all constraints in the inequalities $\phi_a$ are defined with polynomial size binary presentable coefficients, then by classical results on solutions of systems of linear inequalities, see e.g.~\cite{Papadimitriou1981}, there exists a value $v$ for $x$ with a polynomial size binary representation.
\end{proof}

The hardness is established by the following lemma.

\begin{lemma}
There is a polynomial time reduction from the 3SAT problem to the complement of the structural safety problem for one-sided Minkowski games with moves defined by rational closed polytopes.
\end{lemma}
\begin{proof}
Let $\Psi=\{C_1,C_2,\dots,C_n\}$ be a 3SAT instance where each clause $C_i \equiv \ell_{i1} \lor \ell{i2} \lor \ell{i3}$ are literals built from the set of variables $X=\{ x_1,x_2,\dots,x_m \}$. A literal $\ell_{ij}$ is {\em positive} if it is of the form $x$ for some $x \in X$, and it is {\em negative} if it is of the form $\neg x$ for some $x \in X$. We associate with each clause $C_i$, a move $A_i \subseteq \mathbb{R}^{2m}$ and each propositional variable $x_j \in X$ is associated with two dimensions related to real-valued variables $x_{j1},x_{j2}$ in the sequel. The move $A_i$ is defined as the convex hull of the the three vectors $v(\ell_{ij})$ defined as follows:
  $$v(\ell_{ij})(k)=\left\{ 	\begin{array}{ll}
  					0 & \mbox{~if~} k \not= 2i-1 \land k \not= 2i \\
					1 & \mbox{~if~} k= 2i-1 \land \ell_{ij} \mbox{~is positive, or~}k= 2i \land \ell_{ij} \mbox{~is negative}  \\ 	
					-1 & \mbox{~if~} k= 2i-1 \land \ell_{ij} \mbox{~is a negative, or~}k= 2i \land \ell_{ij} \mbox{~is positive}  \\ 	
  					\end{array} \right.$$
					
\noindent
and $\SafeZone$ is defined by the following set of linear constraints:
$$\bigwedge_{x_j \in X} -1 \leq x_{j1} \leq 1 \land -1 \leq x_{j2} \leq 1 \land x_{j1}+x_{j2}=0$$

We proof the correctness of our reduction as follows. {\bf First}, we establish that if $\Psi$ is satisfiable then \playerB\/ wins the one-sided structural safety game that we have constructed.

Let $f : X \rightarrow \{0,1\}$ be a valuation of the propositional variables in $X$ such that $f \models \Psi$. We construct $v_0$ as follows:
  for all $x_j \in X$, $v_0(x_{j1})=1$ if $f(x_j)=1$, and otherwise $v_0(x_{j1})=-1$, and $v_0(x_{j2})=1$ if $f(x_j)=0$, and otherwise $v_0(x_{j2})=-1$.

Now, let us show that for all modes $A_i \in \MoveSetA$, we have that
$$(v_0 \setsum A_i) \cap \overline{\SafeZone} \not= \emptyset.$$

This is the case because $A_i$ is associated with $C_i$. As $f \models \Psi$, we know that there is a literal $\ell_{ij}$ such that $f \models \ell_{ij}$. Assume that $\ell_{ij}=x_k$ (the case for $\neg x_k$ is symmetric). Because $f(x_k)=1$, we have that $v_0(x_{k1})=1$. Now, $A_i$ contains a vertex $a=(0,\dots,0,1,-1,0,\dots,0)$, i.e. $a(k1)=1$, and $a(k2)=-1$. Clearly, $v_0(x_{k1})+a(x_{k1})=2$ and so if \playerB\/ chooses $a \in A_i$, the next position is outside of $\SafeZone$.

{\bf Second}, assume that there is $v$ and $\lambda_2^v : \MoveSetA \rightarrow \mathbb{R}^{2m}$ a vertex strategy of \playerB. This is w.l.o.g. by Lemma~\ref{lem:simple}. Note that we can further assume that $v(x_{k1})\not=0$, and $v(x_{k2})\not=0$ for all $k$, $1 \leq k \leq m$. This is because if $v(x_{k1})=0$ then $v(x_{k_2}=0$ and so by definition of $\SafeZone$ and the moves, it is the case that those two dimensions are not responsible for the violation of safety. So we can assume that all dimension in $v$ are nonzero.

Now, we define $f^v : X \rightarrow \{0,1\}$, $f^v(x_k)=1$ if and only if $v(x_{k1})>0$. Let us now prove that $f^v \models \Psi$. Let $C_i$ be a clause $\ell_{i1} \lor \ell_{i2} \lor \ell_{i3}$. We know that $v \setsum \lambda^v_2(A_i) \not\subseteq \SafeZone$.  This means that there is a vertex $a_{ij}$ of $A_i$ such that $v + a_{ij} \not\in \SafeZone$.	This vertex corresponds to the literal $\ell_{ij}$ and $f^v \models	\ell_{ij}$ by construction of the moves and $\SafeZone$.			
\qed
\end{proof}

\section{Open questions}\label{sec:oq}

By comparing the results from Subsections \ref{subsec:linconstraints} and \ref{subsec:fixeddim}, we see that while deciding the winner in a boundedness Minkowski game is $\textrm{coNP}$-complete in general, it becomes polynomial-time if the dimension of the ambient space is fixed. Thus, it makes a good candidate for an investigation in the setting of parameterized complexity \cite{downeyfellows}. Is the problem fixed-parameter tractable? Is it hard for some $\mathrm{W}[n]$-class?

In Section \ref{sec:undecidable}, we showed that from some dimension $d$ onwards, it becomes undecidable to determine the winner in a safety Minkowski game defined via sets of linear constraints defining open and closed convex polytopes. This gives immediate rise to two questions: first, what happens for small dimensions? Given that our construction needs essentially two dimensions per instruction, and two per counter, an optimal value is presumably obtained by using universal machine having more than 2 counters. Second, what happens if we restrict our attention to games defined via sets of linear constraints that are all non strict (defining closed convex polytopes only)?

\section*{Acknowledgements}
We would like to thank Dr.~Pritha Mahata and Prof.~Jean Cardinal for helpful and stimulating discussions related to a.o. convex polyhedra and linear algebra.

\bibliographystyle{plain}
\bibliography{minkowski}
\end{document}